\documentclass[runningheads]{llncs}

\usepackage{newclude}
\usepackage{graphicx}
\usepackage{amsfonts}
\usepackage{enumerate}
\usepackage[scr]{rsfso}
\usepackage{units}
\usepackage{enumitem}
\usepackage{comment}
\usepackage{color}
\usepackage{textcomp}
\usepackage{hyperref}
\usepackage{multirow}
\usepackage{environ}
\usepackage[boxed]{algorithm2e}
\usepackage{environ}
\usepackage{comment}
\usepackage{makecell}
\usepackage[classfont=roman,funcfont=roman,langfont=roman]{complexity}
\usepackage{tikz}
\usepackage[thmmarks]{ntheorem} 

\newcommand{\MHVfull}{\textsc{Maximum Happy Vertices}}
\newcommand{\MHE}{\textsc{MHE}}
\newcommand{\MHV}{\textsc{MHV}}
\newcommand{\MHEfull}{\textsc{Maximum Happy Edges}}
\newtheorem{rrule}{Reduction rule}

\usepackage{tabularx,lipsum,environ,amsmath,amssymb}
\usepackage{mathtools}

\makeatletter
\newcommand{\problemtitle}[1]{\gdef\@problemtitle{#1}}\newcommand{\probleminput}[1]{\gdef\@probleminput{#1}}\newcommand{\problemquestion}[1]{\gdef\@problemquestion{#1}}\newcommand{\problemparameter}[1]{\gdef\@problemparameter{#1}}
\NewEnviron{problemx}{
	\problemtitle{}\probleminput{}\problemquestion{}\problemparameter{}\BODY \par\addvspace{.5\baselineskip}
	\noindent
	\centering
	\boxed
	{
		\begin{tabularx}{330pt}{@{\hspace{0.5cm}} l p{250pt} c}
			\multicolumn{2}{@{\hspace{0.5cm}}l}{\@problemtitle} \\\textbf{Input:} & \@probleminput \\\ifthenelse{\equal{\@problemparameter}{}}{}{\textbf{Parameter:} & \@problemparameter \\}\textbf{Question:} & \@problemquestion \end{tabularx}}
	\par\addvspace{.5\baselineskip}
}

\newif\ifshort

\ifshort
\makeatletter
\newenvironment{lemmaO}[1][]{
	\if\relax\detokenize{#1}\relax
	\expandafter\@firstoftwo
	\else
	\expandafter\@secondoftwo
	\fi
	{\begin{lemma}[$\star$]}{\begin{lemma}[#1, $\star$]}
		}
		{
		\end{lemma}
	}
	\makeatother
	\makeatletter
	\newenvironment{theoremO}[1][]{
		\if\relax\detokenize{#1}\relax
		\expandafter\@firstoftwo
		\else
		\expandafter\@secondoftwo
		\fi
		{\begin{theorem}[$\star$]}{\begin{theorem}[#1, $\star$]}
			}
			{
			\end{theorem}
		}
		\makeatother
		\makeatletter
		\newenvironment{corollaryO}[1][]{
			\if\relax\detokenize{#1}\relax
			\expandafter\@firstoftwo
			\else
			\expandafter\@secondoftwo
			\fi
			{\begin{corollary}[$\star$]}{\begin{corollary}[#1, $\star$]}
				}
				{
				\end{corollary}
			}
			\makeatother
		\makeatletter
\newenvironment{claimO}[1][]{
	\if\relax\detokenize{#1}\relax
	\expandafter\@firstoftwo
	\else
	\expandafter\@secondoftwo
	\fi
	{\begin{claim}[$\star$]}{\begin{claim}[#1, $\star$]}
		}
		{
		\end{claim}
	}
	\makeatother
			\NewEnviron{proofO}{}
			\NewEnviron{claimproofO}{}
			\else

			\makeatletter
			\newenvironment{proofO}[1][]{\begin{proof}}{\end{proof}}
			\makeatother
			\makeatletter
			\newenvironment{claimproofO}[1][]{\begin{claimproof}}{\end{claimproof}}
			\makeatother
			
			\fi

\renewcommand{\O}{\mathcal{O}}
\newcommand{\Ostar}[1]{\O^*(#1)}

\theoremseparator{.}
\newtheorem{theorem}{Theorem}
\newtheorem{lemma}{Lemma}
\newtheorem{corollary}{Corollary}

\theorembodyfont{\normalfont}
\theoremseparator{.}
\newtheorem{definition}{Definition}
\newtheorem{claim}{Claim}

\theoremstyle{nonumberplain}
\theoremheaderfont{\itshape}
\theoremsymbol{\ensuremath{\Box}}
\theorembodyfont{\normalfont}
\newtheorem{proof}{Proof}
\newtheorem{proofsketch}{Proof sketch}

\makeatletter
\newtheoremstyle{nonumberplainnobrackets}{\item[\theorem@headerfont\hskip\labelsep ##1\theorem@separator]}{\item[\theorem@headerfont\hskip \labelsep ##1\ ##3\theorem@separator]}
\makeatother

\theoremstyle{nonumberplainnobrackets}
\theoremsymbol{\ensuremath{\blacksquare}} 
\theorembodyfont{\normalfont}
\theoremheaderfont{\itshape}
\newtheorem{claimproof}{Proof of Claim \theclaim}

\title{On Happy Colorings, Cuts, and Structural Parameterizations\thanks{This research was supported by the Russian Science Foundation (project 16-11-10123)}}
\author{Ivan Bliznets\inst{1,2} \and
	Danil Sagunov\inst{1}}
\authorrunning{I. Bliznets and D. Sagunov}
\institute{St.\ Petersburg Department of Steklov Institute of Mathematics of the Russian Academy
	of Sciences, St.\ Petersburg, Russia\\
	\email{iabliznets@gmail.com}, \email{danilka.pro@gmail.com} \and National Research University Higher School of 
	Economics, St.\ Petersburg, Russia}

\begin{document}
	\maketitle

\begin{abstract}
We study the \MHVfull~and \textsc{Maximum} \textsc{Happy Edges} problems.
The former problem is a variant of clusterization, where some vertices have already been assigned to clusters.
The second problem gives a natural generalization of \textsc{Multiway Uncut}, which is the complement of the classical \textsc{Multiway Cut} problem.
Due to their fundamental role in theory and practice, clusterization and cut problems has always attracted a lot of attention.
We establish a new connection between these two classes of problems by providing a reduction between \MHVfull~and \textsc{Node Multiway Cut}.
Moreover, we study structural and distance to triviality parameterizations of \MHVfull~and \MHEfull. Obtained results in these directions answer questions explicitly asked in four works: Agrawal\ '17, Aravind et al.\ '16, Choudhari and Reddy\ '18, Misra and Reddy\ '17.
\end{abstract}

\section{Introduction}

In this paper, we study \textsc{Maximum Happy Vertices} and \textsc{Maximum Happy Edges}.
Both problems were recently introduced by Zhang and Li in 2015~\cite{zhang2015algorithmic}, motivated by a study of algorithmic aspects of the homophyly law in large networks.
Informally they paraphrase the law as "birds of a feather flock together".
The law states that in social networks people are more likely to connect with people sharing similar interests with them.
A social network is represented by a graph, where each vertex corresponds to a person in the network, and an edge between two vertices denotes that corresponding persons are connected within the network.
Furthermore, we let vertices have colors assigned.
The color of a vertex indicates type, character or affiliation of the corresponding person in the network.
An edge is called \emph{happy} if its endpoints are colored with the same color.
A vertex is called \emph{happy} if all its neighbours are colored with the same color as the vertex itself.
Equivalently, a vertex is happy if all edges incident to it are happy.
The formal definitions of \textsc{Maximum Happy Vertices} and \textsc{Maximum Happy Edges} are the following:

\begin{problemx}
	\problemtitle{\MHVfull~(\MHV)}
	\probleminput{A graph $G$, a partial coloring of vertices $p: S \rightarrow [\ell]$ for some $S\subseteq V(G)$ and an integer $k$.}
	\problemquestion{Is there a coloring $c: V(G) \rightarrow [\ell]$ extending the partial coloring $p$ such that the number of happy vertices with respect to $c$ is at least $k$?}
\end{problemx}

\begin{problemx}
	\problemtitle{\MHEfull~(\MHE)}
	\probleminput{A graph $G$, a partial coloring of vertices $p: S \rightarrow [\ell]$ for some $S\subseteq V(G)$ and an integer $k$.}
	\problemquestion{Is there a coloring $c: V(G) \rightarrow [\ell]$ extending the partial coloring $p$ such that the number of happy edges with respect to $c$ is at least $k$?}
\end{problemx}

\MHEfull~has an immediate connection to \textsc{Multiway Cut}.
Precisely, if each color is used in precoloring exactly once, then \textsc{Maximum Happy Edges} is exactly the \textsc{Multiway Uncut} problem, i.e.\ the edge complement of \textsc{Multiway Cut}.
Thus, \textsc{Maximum Happy Edges} is a generalization of the \textsc{Multiway Uncut} problem.
So, in this case the connection between clustering vertices by color and cutting edges in order to separate different colors is pretty obvious.
However, this is not the case for vertex version of the problem, which we would like to connect with the vertex version of \textsc{Multiway Cut}, \textsc{Node Multiway Cut}.

\MHVfull~can be seen as a sort of clusterization problem, in which some vertices already have prescribed color/cluster and the goal is to identify colors/clusters of initially uncolored/unassigned vertices.
In some sense, we would like to clusterize the graph in such a way that overall boundary of clusters is minimized.
Here, by a boundary of a cluster we understand vertices of the cluster that are connected to vertices outside the cluster.
While it is possible to straightforwardly formulate the problem in terms of a special cutting problem, this kind of formalization will sound complicated and unnatural.
We show that \MHV~can be easily transformed into \textsc{Node Multiway Cut},
thereby constructing an additional bridge between clusterization and cutting problems.

Recently, \MHV~and \MHE~have attracted a lot of attention and were studied from parameterized~\cite{Agrawal2018,Aravind2016,Aravind2017,Choudhari2018,Misra2018} and  approximation~\cite{zhang2015algorithmic,zhang2018improved, zhang2015improved,xu2016submodular} points of view as well as from experimantal perspective~\cite{lewis2019finding}.
Further, dozens of algorithms for the classical \textsc{Multiway Cut} problem have been considered as well, which is the complement of a special case of \MHE.

In 2015, Zhang and Li established that $\ell$-\MHE~and $\ell$-\MHV~are \NP-hard for $\ell\geq3$, where $\ell$ is the number of colors used.
Later, Aravind et al.\cite{Aravind2016} showed that when the input graph is a tree, $\ell$-\MHV~and $\ell$-\MHE~can be solved in $\O(n\ell \log\ell)$ and in $\O(n\ell)$ time respectively.
In \cite{Misra2018}, Misra and Reddy proved \NP-hardness of both \MHV~and \MHE~on split and on bipartite graphs, and showed that \MHV~is polynomial time solvable on cographs.

From the approximation perspective, the currently best known results are the following.
Zhang et al.~\cite{zhang2018improved} showed that \MHV~can be approximated within $\frac{1}{\Delta + 1}$, where $\Delta$ is the maximum degree of the input graph, and \MHE~can be approximated within $\frac{1}{2}+\frac{\sqrt{2}}{4}f(\ell)$, where $f(\ell)= \frac{(1-1/\ell)\sqrt{\ell(\ell-1)}+1/\sqrt{2}}{\ell-1+1/2\ell} \leq 1$.
They also claimed that a more careful analysis can improve the approixmation ratio for \MHV~to $\frac{1}{\Delta+1/g(\Delta)}$, where $g(\Delta)=(\sqrt{\Delta}+\sqrt{\Delta+1})^2\Delta>4 \Delta^2$. 

The known results in parameterized complexity (not including kernelization) are summarized  in Table~\ref{table:results}.
Results proved in the paper are marked by $^*$ in the table.
Agrawal~\cite{Agrawal2018} provides $\mathcal{O}(k^2\ell^2)$-kernel for MHV, where $\ell$ is the number of used colors and $k$ is the number of desired happy vertices.
Independently, Gao and Gao~\cite{gao2018kernelization} present a $(2^{k \ell + k} + k\ell +k+\ell)$-kernel for the general case and a $(7(k\ell+k)+\ell-10)$-kernel in the case of planar graphs.
We provide a kernel on $\O(d^3)$ vertices for \MHV~parameterized by the distance to clique, partially answering a question in \cite{Misra2018}.
Note that the kernel sizes mentioned in this paragraph correspond to the number of vertices in the kernels.

\begin{table}[!ht]
	\centering

\begin{tabular}{|m{2.65cm}|c|c|c|c|c|c|}
	\hline
	\textbf{Parameter} & \textbf{\MHE} & \textbf{$\ell$-\MHE} & \textbf{\MHV} & \textbf{$\ell$-\MHV} \\
\hline
	
	Distance to threshold graphs & ?& ?& \multicolumn{2}{c|}{$d^{\O(d)} \cdot n^{\O(1)}$ \cite{Choudhari2018}} \\
	\hline
	Distance to clique & \multicolumn{4}{c|}{$d^{\O(d)}\cdot n^{\O(1)}$ \cite{Misra2018}} \\ 
	\hline
	Distance to cluster &  $\W[1]$-hard C\ref{cor:mhe_w1_cluster}$^*$ & $\ell^d \cdot n^{\O(1)}$ &  \multicolumn{2}{c|}{$d^{\O(d)} \cdot n^{\O(1)}$ T\ref{thm:mhv_cluster_fpt}$^*$} \\

    \hline
    Distance to cographs & \multirow{4}{*}{$\W[1]$-hard C\ref{cor:mhe_w1_structural}$^*$} & ? & \multirow{4}{*}{$\W[1]$-hard C\ref{cor:mhv_w1_structural}$^*$} & $(2\ell)^{d} \cdot n^{\O(1)}$ \\
   
\cline{1-1}\cline{3-3}\cline{5-5}
Treewidth &  & $\ell^{\text{tw}} \cdot n^{\O(1)}$ \cite{Aravind2017,Agrawal2018} &  & $\ell^{\text{tw}} \cdot n^{\O(1)}$ \cite{Aravind2017,Misra2018} \\
	\cline{1-1}\cline{3-3}\cline{5-5}
	Pathwidth &  & $\ell^{\text{pw}} \cdot n^{\O(1)}$ \cite{Aravind2017,Misra2018} &  & $\ell^{\text{pw}} \cdot n^{\O(1)}$ \cite{Aravind2017,Agrawal2018} \\
	\cline{1-1}\cline{3-3}\cline{5-5}
	Cliquewidth & & ? &  & ? \\
	\cline{1-1}\cline{3-3}\cline{5-5}
	\hline
	Feedback Vertex Set Number &  & $\ell^d \cdot n^{\O(1)}$ &  & $(2\ell)^d \cdot n^{\O(1)}$ \\
	 	\hline
 Vertex Cover Number  & \multicolumn{4}{c|}{$d^{\O(d)} \cdot n^{\O(1)}$ \cite{Misra2018}} \\
	\hline
Split Vertex Deletion Number & \multicolumn{4}{c|}{\multirow{3}{*}{para-NP-hard \cite{Misra2018}}} \\
\cline{1-1}
Odd Cycle Transversal Number & \multicolumn{4}{c|}{} \\

	\hline
	
	Neighbourhood Diversity &
	\multicolumn{4}{c|}{$2^{\text{nd}}\cdot n^{\O(1)}$\cite{Aravind2017}} \\
\hline
\end{tabular}

\caption{Known and established results under distance-to-triviality and structural parameters. 
$^*$ marks results of this work.
T indicates a result proven as a theorem.
C indicates a result proven as a corollary.
$d$ denotes the distance parameter of the row.}
\label{table:results}
\end{table}
\clearpage{}

\textbf{Our results:} The main contributions of our work are the following.

\begin{itemize}
	\item We establish a natural connection between \MHVfull~on a graph $G$ and \textsc{Node Multiway Cut} on a second power of a certain subgraph of $G$.
	
	\item We answer questions in~\cite{Agrawal2018,Aravind2016} about existence of \FPT-algorithm for \MHV \linebreak parameterized by the treewidth of the input graph only.

	\item 
Similarly, we answer one of the questions from Choudhari et al.~\cite{Choudhari2018} and Misra et al.~\cite{Misra2018} by showing $\W[1]$-hardness of \MHE~parameterized by the cluster vertex deletion number.
	We show that \MHV, in contrast to \MHE, is in \FPT~when parameterized by the cluster vertex deletion number.
	
	\item We partially answer a question stated by Misra and Reddy in \cite{Misra2018}.
	We provide a kernel of size $\O(d^3)$ for \MHV, where $d$ is the distance to cliques. 

	\item Among other results, we also present the first algorithm for \textsc{Node Multiway Cut} parameterized by the clique-width of the input graph.
	
\end{itemize}

\textbf{Organization of the paper:}
Section~\ref{sec:distance-param} describes results under some structural and distance-to-triviality parameters.
In Section~\ref{section:cut-relation} we provide results connecting \textsc{Node Multiway Cut} and \textsc{Maximum Happy Vertices}. 
In Section~\ref{sec:structural-kernel} we provide a polynomial kernel for \MHV~parameterized by the distance to clique.
In Section~\ref{sec:w2} we show how to strengthen the results of \W[1]-hardness and obtain the corresponding \W[2]-hardness results.

	\section{Preliminaries}

\textbf{Basic notation.} We denote the set of positive integer numbers by $\mathbb{N}$.
For each positive integer $k$, by $[k]$ we denote the set of all positive integers not exceeding $k$, $\{1,2,\ldots, k\}$.
We use $\infty$ to denote an infinitely large number, for which holds $n<\infty$ and $n+\infty=\infty+n=\infty$, where $n$ is an arbitrary integer.
We use $\sqcup$ for the disjoint union operator, i.e.\ $A\sqcup B$ equals $A\cup B$, with an additional constraint that $A$ and $B$ are disjoint.

\ifshort
\else
We employ partial functions in our work.
To denote a \emph{partial function} from a set $X$ to a set $Y$, that is, a function that may do not map some element of $X$ to an element of $Y$, we write $f: X \nrightarrow Y$.
If a partial function $f$ maps an element $x \in X$ to some element in $Y$, we say that $f(x)$ is \emph{assigned}.
If $f(x)$ is unassigned, we allow to extend $f$ by assigning the value of $f(x)$.
\fi

We use the traditional $\O$-notation for asymptotical upper bounds.
We additionally use the $\mathcal{O}^*$-notation that hides polynomial factors.
Many of our results concern the parameterized complexity of the problems, including fixed-parameter tractable algorithms, kernelization algorithms, and some hardness results for certain parameters.
For a detailed survey in parameterized algorithms we refer to the book of Cygan et al.\ \cite{cygan2015parameterized}.
In their book one may also find definitions of pathwidth and treewidth that are considered as parameters in some of our results.

Throughout the paper, we use standard graph notation and terminology, following the book of Diestel \cite{diestel2018graph}.
All graphs in our work are undirected simple graphs.
We consider several graph classes in our work.
\emph{Interval graphs} are graphs whose vertices can be represented as intervals on the real line, so that a pair of vertices are connected by an edge if and only if their representative intervals intersect.
\emph{Cluster graphs} are graphs that are a disjoint union of cliques, or, equivalently, graphs that do not contain induced paths on three vertices.

We often refer to the distance to $\mathcal{G}$ parameter, where $\mathcal{G}$ is an arbitrary graph class.
For a graph $G$, we say that a vertex subset $S\subseteq V(G)$ is a \emph{$\mathcal{G}$ modulator} of $G$, if $G$ becomes a member of $\mathcal{G}$ after deletion of $S$, i.e.\ $G\setminus S \in \mathcal{G}$.
Then, the \emph{distance to $\mathcal{G}$} parameter of $G$ is defined as the size of its smallest $\mathcal{G}$ modulator.

\textbf{Graph colorings.} When dealing with instances of \MHVfull~or \MHEfull, we use a notion of colorings.
A \textit{coloring} of a graph $G$ is a function that maps vertices of the graph to a set of colors.
If this function is partial, we call such a coloring \emph{partial}.
If not stated otherwise, we use $\ell$ for the number of distinct colors, and assume that colors are integers in $[\ell]$.
A partial coloring $p$ is always given as a part of the input for both problems, along with graph $G$.
We also call $p$ a \emph{precoloring} of the graph $G$, and use $(G,p)$ to denote the graph along with the precoloring.
The goal of both problems is to extend this partial coloring to a specific coloring $c$ that maps each vertex to a color.
We call $c$ a \emph{full coloring} (or simply, a coloring) of $G$ that extends $p$.
We may also say that $c$ is a coloring of $(G,p)$.
For convenience, introduce the notion of potentially happy vertices, both for full and partial colorings.

\begin{definition}
	We call a vertex $v$ of $(G,p)$ \textit{potentially happy}, if there exists a coloring $c$ of $(G,p)$ such that $v$ is happy with respect to $c$. In other words, if $u$ and $w$ are precolored neighbours of $v$, then $p(u)=p(w)$ (and $p(u)=p(v)$, if $v$ is a precolored vertex). We denote the set of all potentially happy vertices in $(G,p)$ by $\mathcal{H}(G,p)$.
	
	By $\mathcal{H}_i(G,p)$ we denote the set of all potentially happy vertices in $(G,p)$ such that they are either precolored with color $i$ or have a neighbour precolored with color $i$: $$\mathcal{H}_i(G,p)=\{v \in \mathcal{H}(G,p) \mid N[v] \cap p^{-1}(i)\neq \emptyset\}.$$
	In other words, if a vertex $v \in \mathcal{H}_i(G,p)$ is happy with respect to some coloring $c$ of $(G,p)$, then necessarily $c(v)=i$.
\end{definition}

Note that if $c$ is a full coloring of a graph $G$, then $|\mathcal{H}(G,c)|$ is equal to the number of vertices in $G$ that are happy with respect to $c$.

\textbf{Clique-width.} Among other structural parameters, we consider clique-width in our work.
We follow definitions presented by Lackner et al.\ in their work on \textsc{Multicut} parameterized by clique-width \cite{Lackner2012}.

\ifshort
\else
To define clique-width, we need to define $k$-expressions first.
For any $k \in \mathbb{N}$, a \emph{$k$-expression} $\Phi$ describes a graph $G_\Phi$, whose vertices are labeled with integers in $[k]$.
$k$-expressions and its corresponding graphs are defined recursively.
Depending on its topmost operator, a $k$-expression $\Phi$ can be of four following types.

\begin{enumerate}
	\item \emph{Introducing a vertex.}
	$\Phi=i(v)$, where $i \in [k]$ is a label and $v$ is a vertex.
	$G_\Phi$ is a graph consisting of a single vertex $v$ with label $i$, i.e.\ $V(G_\Phi)=\{v\}$.
	\item \emph{Disjoint union.}
	$\Phi=\Phi'\oplus\Phi''$, where $\Phi'$ and $\Phi''$ are smaller subexpressions.
	$G_\Phi$ is a disjoint union of the graphs $G_{\Phi'}$ and $G_{\Phi''}$, i.e.\ $V(G_\Phi)=V(G_{\Phi'})\sqcup V(G_{\Phi''})$ and $E(G_\Phi)=E(G_{\Phi'})\sqcup E(G_{\Phi''})$.
	The labels of the vertices remain the same.
	\item \emph{Renaming labels.}
	$\Phi=\eta_{i\to j}(\Phi')$.
	The structure of $G_\Phi$ remains the same as the structure of $G_{\Phi'}$, but each vertex with label $i$ receives label $j$.
	\item \emph{Introducing edges.}
	$\Phi=\rho_{i,j}(\Phi')$.
	$G_\Phi$ is obtained from $G_{\Phi'}$ by connecting each vertex with label $i$ with each vertex with label $j$.	
\end{enumerate}

\emph{Clique-width} of a graph $G$ is then defined as the smallest value of $k$ needed to describe $G$ with a $k$-expression and is denoted as $\operatorname{cw}(G)$.
To avoid confusion with the parameter $k$ of \MHV~and \MHE, we may use notation of $w$-expression instead of $k$-expression.

There is still no known \FPT-algorithm for finding a $k$-expression of a given graph $G$.
However, there is an \FPT-algorithm that decides that $\operatorname{cw}(G) > k$ or outputs $(2^{3k+2}-1)$-expression of $G$.
\fi
For more details on clique-width we refer to \cite{Hlineny2007}.

	\section{Structural and distance-to-triviality parameters}\label{sec:distance-param}

In \cite{Agrawal2018}, Agrawal proved that \MHVfull~is $\W[1]$-hard with respect to the standard parameter, the number of happy vertices. In \cite{Aravind2016,Misra2018,Choudhari2018} some structural parameters for \textsc{MHV} and \textsc{MHE} were studied.
In \cite{Agrawal2018}, Agrawal also asked whether \MHV~admits an \FPT~algorithm when parameterized by the treewidth of the input graph alone.
In this section, we show that both \MHV~and \MHE~are \W[1]-hard with respect to certain distance-to-triviality and structural paramters, including treewidth, answering the question of Agrawal and some other questions.
We start with the definition of a classical $\W[1]$-complete (with respect to the solution size) problem.

\begin{problemx}
	\problemtitle{\textsc{Regular Multicolored Independent Set}}
	\probleminput{Graph $G$, with degree of every vertex in $G$ equal to $r$, a partition of $G$ into $k$ cliques $V_1, V_2, \ldots, V_k$.}
	\problemparameter{$k$}
	\problemquestion{Is there a multicolored independent set in $G$ of size $k$, i.e.\ a subset $S\subseteq V(G)$ of its vertices that is an independent set in $G$ and $|S\cap V_i|=1$ for every $i \in [k]$?}
\end{problemx}

\begin{theorem}\label{thm:mhv_p3}
	\MHVfull~is $\W[1]$-hard when parameterized by the distance to graphs that are a disjoint union of paths consisting of three vertices.
\end{theorem}
\begin{proof}
	We reduce from \textsc{Regular Multicolored Independent Set}, that is $\W[1]$-complete with respect to $k$ due to \cite{Belmonte2013}.
	
	Let $(G, k, V_1, V_2, \ldots, V_k)$ be an instance of \textsc{Regular Multicolored Independent Set}, and let $r$ be the degree of every vertex in $G$, i.e.\ $r=|N(v)|$ for any $v \in V(G)$. We assume that $|V_i| \ge 2$ for each $i$, since otherwise the instance can be trivially reduced to an instance with a smaller $k$. We construct an instance $(G', p, k')$ of \MHVfull~as follows.
	
	We set $\ell=|V(G)|$, so each color corresponds to a unique vertex of $G$. For convenience, we use vertices of $G$ as colors, instead of the numbers in $[\ell]$.
	
	For each edge $uv \in E(G)$, we introduce a path on three vertices $t_{uv}^{u}$, $e_{uv}$, $t_{uv}^v$ in $G'$, with $e_{uv}$ being the middle vertex of the path. Endpoint vertices $t_{uv}^u$ and $t_{uv}^v$ are precolored in colors $u$ and $v$ respectively, i.e.\ $p(t_{uv}^u)=u$ and $p(t_{uv}^v)=v$, and the middle vertex is left uncolored.
	
	We then introduce a selection gadget in $G'$. That is, we introduce $k$ uncolored vertices $s_1, s_2, \ldots, s_k$. For each $i \in [k]$ and each color $v \in V_i$, we connect $s_i$ with each vertex precolored in color $v$. Thus, a vertex $t_{uv}^u$ becomes connected to exactly one vertex of the selection gadget $s_i$, where $i$ is such that $u \in V_i$. The purpose of the selection gadget is that the color of $s_i$ in the optimal coloring corresponds to a vertex that we should take in $V_i$ in the initial instance of \textsc{Regular Multicolored Independent Set}.
	
	We finally set $k'=kr$ and argue that $(G, k, V_1, V_2, \ldots, V_k)$ is a yes-instance of \textsc{Regular Multicolored Independent Set} if and only if $(G', p, k')$ is a yes-instance of \MHVfull.

	Let $S \subseteq V(G)$ be a multicolored independent set of $G$, i.e.\ $S$ is an independent set in $G$ and $|S \cap V_i|=1$ for each $i$. Let us construct a coloring $c$ of $V(G')$ such that it extends $p$ and at least $k'=kr$ vertices of $G$ are happy with respect to $c$. For each $i$, set the color of $s_i$ to $v_i$, i.e.\ $c(s_i)=v_i$, where $v_i \in S \cap V_i$. For each edge $uv \in E(G)$, set the color of $e_{uv}$ to $u$, if $u \in S$, or to $v$, if $v \in S$, and to an arbitrary color otherwise. Formally, $c(e_{uv})=u$, if $u \in S$, and $c(e_{uv})=v$, if $v \in S$. If $u,v \notin S$, then $c(e_{uv})$ can be assigned an arbitrary color.
	Note that either $u \notin S$ or $v \notin S$, since $S$ is an independent set.
	
	$G'$ has no other uncolored vertex, thus the construction of $c$ is complete. 
	
	\begin{claim}
		For each vertex $u \in S$ and each edge $uv \in E(G)$ incident to $u$, $t_{uv}^u$ is happy with respect to $c$.
	\end{claim}

	\begin{claimproof}
		Indeed, $t_{uv}^u$ is adjacent to exactly two vertices: $s_i$, where $u \in V_i$, and $e_{uv}$. Since $u \in S$, $c(s_i)=u$ and $c(e_{uv})=u$ by construction of $c$. $t_{uv}^u$ is a vertex precolored with color $u$, hence $t_{uv}^u$ is happy with respect to $c$.
	\end{claimproof}
	
	For each $u \in S$, there are exactly $r$ edges adjacent to $u$, hence all $r$ vertices precolored with color $u$ are happy. $|S|=k$, hence at least $kr$ vertices of $G'$ are happy with respect to $c$.
	
	It is left to prove that if $(G', p, k')$ is a yes-instance of \MHVfull, then $(G, k, V_1, V_2,$ $ \ldots, V_k)$ is a yes-instance of \textsc{Regular Multicolored Independent Set}.
		
	\begin{claim}\label{claim:mhv_kr}
		Let $c$ be an arbitrary coloring of $V(G')$ extending $p$. There are at most $kr$ happy vertices in $G'$ with respect to $c$. Moreover, all happy vertices are precolored vertices of at most $k$ distinct colors.
	\end{claim}

	\begin{claimproof}

		Observe that for each $i \in [k]$, $s_i$ is unhappy with respect to any coloring extending $p$, since neighbours of $s_i$ are precolored with colors in $V_i$, and each color is presented exactly $r$ times among its neighbours, and we assumed that $V_i$ consists of at least two vertices.
		
		For each $uv \in E(G)$, $e_{uv}$ is adjacent to exactly two vertices $t_{uv}^u$ and $t_{uv}^v$, which are precolored with two distinct colors $u$ and $v$. Thus, $e_{uv}$ is also unhappy with respect to any coloring extending $p$.
	
		Hence, only precolored vertices of $G'$ can be happy, i.e.\ vertices $t_{uv}^u$ for $uv \in E(G)$. Each of them is adjacent to exactly one vertex of the selector gadget, i.e.\ vertex $s_i$ for some $i \in [k]$. But for each $i \in [k]$, only the neighbours that share the same color as $s_i$ can be happy. Thus, each happy vertex shares a color with one of $k$ vertices of the selection gadget. Since each color is presented exactly $r$ times in the partial coloring $p$, there can be at most $kr$ such happy vertices.
	\end{claimproof}

	Let $c$ be a coloring of $V(G')$ extending $p$ such that at least $kr$ vertices of $G'$ are happy with respect to $c$. According to Claim \ref{claim:mhv_kr}, exactly $kr$ vertices of $G'$ are happy with respect to $c$, and they are precolored with $k$ different colors. Moreover, for each color, all $r$ precolored vertices of this color are happy. Let $S$ be the set of these $k$ colors, i.e.\ $S=\{c(s_1), c(s_2), \ldots, c(s_k)\}$. We argue that $S$ is an independent set in $G$. Note that $|S \cap V_i|=1$ is then automatically satisfied, as $V_i$ is a clique in $G$ for each $i \in [k]$.
	
	\begin{claim}\label{claim:mhv_constructing_is}
		If there are $kr$ happy vertices among the vertices of type $t_{uv}^v$ in $G'$ with respect to coloring $c$, that extends $p$, then $S=\{c(s_1), c(s_2), \ldots, c(s_k)\}$ is an independent set in $G$.
	\end{claim}
	\begin{claimproof}
		Indeed, suppose that $S$ is not an independent set in $G$, i.e.\ there are vertices $u, v \in S$, such that $uv \in E(G)$. Then there is a path $t_{uv}^u$, $e_{uv}$, $t_{uv}^v$ in $G'$. $t_{uv}^u$ is a happy vertex of color $c(t_{uv}^u)=p(t_{uv}^u)=u$, hence $c(e_{uv})=u$. Analogously, $t_{uv}^v$ is a happy vertex of color $v$, hence $c(e_{uv})=v$. We get that $u=c(e_{uv})=v$, which contradicts our assumption. 
	\end{claimproof}
	
	We have shown that $(G', p, k')$ is an instance equivalent to $(G, k, V_1, V_2, \ldots,$ $ V_k)$; moreover, it can be constructed in polynomial time.
	
	Note that the deletion of the selector gadget vertices in $G'$ leads to $G'$ being a disjoint union of paths consisting of three vertices. Thus, $G'$ has the distance parameter being at most $k$, and if \MHVfull~is in $\FPT$ when parameterized by the distance to graphs being a disjoint union of path consisting of three vertices, then $\W[1]$-complete \textsc{Regular Multicolored Independent Set} is also in $\FPT$. Hence, \textsc{MHV} is $\W[1]$-hard with respect to the distance parameter.
\end{proof}

The following corollary answers an open question posed in \cite{Agrawal2018}.

\begin{corollary}\label{cor:mhv_w1_structural}
	\MHVfull~is $\W[1]$-hard with respect to parameters pathwidth, treewidth or clique-width, distance to cographs, feedback vertex set number.
\end{corollary}
\ifshort
\else
\begin{proof}
	$\W[1]$-hardness of \MHV~with respect to the parameters distance to cographs or feedback vertex set number is an immediate corollary of Theorem \ref{thm:mhv_p3}, since graphs of type $n \times P_3$ (that is, graphs that are a disjoint union of paths consisting of three vertices) are simultaneously cographs and forests.
	
	\emph{Pathwidth}. Let $G$ be a graph and $S \subseteq V(G)$ be a $n \times P_3$ modulator of $G$, i.e.\ $G[V(G) \setminus S]$ is a graph consisting of connected components that are disjoint paths on three vertices. Observe that the pathwidth of $G$ is at most $|S|+1$. Indeed, let $V(G) \setminus S$ consist of $n$ connecting components, $i^{\text{th}}$ of them is a three-vertex path $v_{i,1}-v_{i,2}-v_{i,3}$. Then construct a path decomposition of $G$ as a sequence
	$S \cup \{v_{1,1}, v_{1,2}\}, S \cup \{v_{1,2}, v_{1,3}\}, S \cup \{v_{2,1}, v_{2,2}\}, \ldots, S \cup \{v_{i,1}, v_{i,2}\}, S \cup \{v_{i,2}, v_{i,3}\},\ldots, S\cup\{v_{n,2}, v_{n,3}\}.$
	
	Constructed sequence is a correct path decomposition of $G$. Firstly, each vertex is contained in a contiguous segment of sets in the sequence. Secondly, for each edge in $E(G)$, its endpoints are contained in some set of the sequence simultaneously, as each edge of $G$ is either an edge between a vertex in $S$ and some vertex $v_{i,j}$, or an edge between $v_{i,t}$ and $v_{i,t+1}$ for some $i \in [n]$ and $t \in [2]$. The size of each set of the sequence is $|S|+2$, hence the pathwidth of $G$ is at most $|S|+1$. Thus, if a graph has the distance-to-$n\times P_3$ graphs parameter equal to $k$, then its pathwidth is at most $k+1$. By Theorem \ref{thm:mhv_p3}, \MHV~is $\W[1]$-hard when parameterized by the pathwidth of the input graph.
	
	\emph{Treewidth}. $\W[1]$-hardness for the treewidth parameter follows from the fact that a path decomposition of a graph is a tree decomposition of the graph; if a graph is of pathwidth $k$, it is of treewidth at most $k$.
	
	\emph{Clique-width}. In \cite{Corneil2005}, Corneil and Rotics proved that a graph of treewidth $k$ has clique-width at most $3 \cdot 2^{k-1}$. This already gives us the hardness result for the clique-width parameter. Though, one can improve the upper bound and show that if a graph has a $n \times P_3$-modulator of size $k$, then the clique-width of such graph is at most $k+3$.
\end{proof}
\fi

\begin{theorem}\label{thm:mhe_p3_c3}
	\MHEfull~is $\W[1]$-hard when parameterized by the distance to graphs that are disjoint union of paths consisting of three vertices and is $\W[1]$-hard when parameterized by the distance to graphs that are a disjoint union of cycles of length three.
\end{theorem}
\ifshort
\else
\begin{proof}
	We adjust the reduction from \textsc{Regular Multicolored Independent Set} to \textsc{MHV} provided in the proof of Theorem \ref{thm:mhv_p3}.
	
	Given an instance $(G, k, V_1, V_2, \ldots, V_k)$ of \textsc{Regular Multicolored Independent Set}, we construct an instance $(G', p, k')$ of \MHEfull~as follows.
	
	Let $n=|V(G)|$, $m=|E(G)|$. $G'$ is constructed in the same way as in the proof of Theorem \ref{thm:mhv_p3}: for each edge $uv \in E(G)$, we introduce a path on three vertices $t_{uv}^u$, $e_{uv}$, $t_{uv}^v$, and set $p(t_{uv}^u)=u$, $p(t_{uv}^v)=v$, and $e_{uv}$ is left uncolored; then we introduce the selection gadget vertices $s_1, s_2, \ldots, s_k$, and introduce an edge between $s_i$ and $t_{uv}^u$ for each $i \in [k]$, $u \in V_i$ and $uv \in E(G)$. For each $i \in [k]$, $s_i$ is left uncolored.
	
	Additionally, we introduce edges new to this construction: for each $i, j \in [k]$ and each edge $uv \in E(G)$, such that $u \in V_i$ and $v \in V_j$, we introduce edges between $e_{uv}$ and $s_i$ and between $e_{uv}$ and $s_j$. In case $i=j$, we introduce only one edge.
	
	We also need additional precolored vertices in order for this reduction to work. For each $i \in [k]$, and each $v \in V_i$, we introduce $m$ new paths consisting of three vertices in $G'$: for each $j \in [m]$, we introduce a path $a_{v,j}^1$, $a_{v,j}^2$, $a_{v,j}^3$. We precolor every vertex in these new paths with color $v$, i.e.\ $p(a_{v,j}^1)=p(a_{v,j}^2)=p(a_{v,j}^3)=v$ for each $j$.
	Then we connect each of them by a newly-introduced edge to the vertex $s_i$ of the selector gadget. These auxiliary vertices will ensure that for each $i \in [k]$, $s_i$ is colored with one of the colors in $V_i$. Note that paths between these newly-introduced vertices $a^t_{v,j}$ are needed only to preserve the distance parameter.
	
	\begin{claim}\label{claim:mhe_aux_ens}
		In any optimal coloring $c$ of $G'$ extending $p$, $c(s_i) \in V_i$ for each $i \in [k]$.
	\end{claim}

	\begin{claimproof}
		Suppose $c$ is an optimal coloring of $G'$ extending $p$, but $c(s_i) \notin V_i$ for some $i \in [k]$. Then no edge between $s_i$ and $a_{v,j}^k$ is happy for any $j, k$ with respect to $c$. Hence, the only edges incident to $s_i$ that can be happy are edges between $s_i$ and vertices of the paths constructed for edges, i.e.\ $t_{uv}^u$ or $e_{uv}$. There are exactly $3m$ such vertices, thus $s_i$ is incident to at most $3m$ edges happy with respect to $c$. But for each $v \in V_i$, $s_i$ is adjacent to $3m$ vertices of type $a_{v,j}^k$ and $r$ vertices of type $t_{uv}^{v}$ precolored with color $v$ . Hence, if we change the color of $s_i$ in $c$ to one of the colors in $V_i$, we lose at most $3m$ happy edges, and win at least $3m+r$ happy edges, which contradicts the optimality of $c$.
	\end{claimproof}
	
	We finally set $k'=kr+(m+kr)+(3k+2n)\cdot m$ and argue that $(G, k, V_1, V_2,\ldots,$ $ V_k)$ is a yes-instance of \textsc{Regular Multicolored Independent Set} if and only if $(G', p, k')$ is a yes-instance of \MHEfull.
	
	Again, similarly to the proof of Theorem \ref{thm:mhv_p3}, let construct a coloring $c$ of $G'$ from a multicolored independent set $S$ of $G$ with $|S|=k$. The coloring is constructed almost in the same way as in the proof of Theorem \ref{thm:mhv_p3}: for each $i \in [k]$, we put $c(s_i)=v_i$, where $v_i \in S \cap V_i$, and for each $uv_i \in E(G)$ we put $c(e_{uv_i})=v_i$. The difference is in coloring vertices $e_{uv}$, where $u \notin S$ and $v \notin S$. Since $e_{uv}$ is now adjacent to one or two vertices of the selector gadget, one may win one happy edge by coloring $e_{uv}$ with one of the colors of the selector gadget vertices. Thus we put
	$$c(e_{uv})=\left\{
		\begin{matrix}
			u, && \text{if $u \in S$,} \\
			v, &&\text{if $v \in S$,} \\
		c(s_i)\text{ or }c(s_j), && \text{where $u \in V_i$ and $v \in V_j$, otherwise.}
		\end{matrix}
	\right.$$
	
	\begin{claim}\label{claim:mhe_s_to_c}
		There are exactly $k'=kr+(m+kr)+(3k+2n)\cdot m$ edges that are happy with respect to $c$.
	\end{claim}
	\begin{claimproof}
		Let consider every type of edges in $G'$.
		\begin{enumerate}
			\item Edges inside of the path $a_{v,j}^1, a_{v,j}^2, a_{v,j}^3$ for any $v \in V(G)$ and $j \in [m]$.
			
			Each such path gives exactly $2$ happy edges, and there are $nm$ such paths. In total, there are $2nm$ edges of this type.
			
			\item Edges of type $(s_i, a_{v,j}^k)$, for any $i \in [k]$, $v \in V_i$, $j \in [m]$, and $k \in [3]$.
			
			Since $c(s_i)=v_i \in V_i$, only edges between $s_i$ and $a_{v_i, j}^k$ are happy for a fixed $i$. There are $k \cdot m \cdot 3$ possible options of choosing $i, j$ and $k$, hence these edges are $3km$ in total.
			
			\item Edges between $s_i$ and $t_{uv}^v$ for any $i \in [k]$, $v \in V_i$ and $uv \in E(G)$.
			
			Again, since $t_{uv}^v$ is precolored with color $v$, and $c(s_i)=v_i \in V_i$, only edges between $s_i$ and $t_{uv_i}^{v_i}$ are happy. There are exactly $r$ edges incident to $v_i$ in $G$, hence $s_i$ is adjacent to exactly $r$ vertices of type $t_{uv_i}^{v_i}$. In total, these sum up to $kr$ edges.
			
			\item Edges between $e_{uv}$ and $t_{uv}^v$ for any $uv \in E(G)$.
			
			Since $c(e_{uv})=c(t_{uv}^v)$ if and only if $v \in S$, for each fixed $v \in S$ there are exactly $r$ happy edges of such type. Hence, there are $kr$ such happy edges.
			
			\item Edges between $s_i$ and $e_{uv}$, for any $i \in [k]$, $v \in V_i$ and $uv \in E(G)$, where $u \notin S$ and $v \notin S$.
			
			For each such $uv \in E(G)$, we constructed $c$ so that $e_{uv}$ is colored in the color of one of its neighbours in the selector gadget. $e_{uv}$ is adjacent to one or two selector gadget vertices of distinct colors, hence it is adjacent to exactly one edge of such type. There are exactly $m-kr$ edges in $G$ with no endpoints in $S$, thus exactly $m-kr$ edges of type $(s_i, e_{uv})$ are happy in $G'$ with respect to $c$.
			
			\item Edges between $s_i$ and $e_{uv}$, for any $i \in [k]$, $v \in V_i$ and $uv \in E(G)$, where $v \in S$.
			
			As $v \in S$, $v \in S \cap V_i$, hence $c(s_i)=v$. Also, since $v \in S$, $c(e_{uv})=v$. Thus, each edge of such type in $G'$ is happy with respect to $c$, and there are $kr$ edges of such type.
		\end{enumerate}
	
		In total, we get that exactly
		$$2nm+3km+kr+kr+(m-kr)+kr=k'$$
		edges are happy in $G'$ with respect to $c$.
	\end{claimproof}

	Claim \ref{claim:mhe_s_to_c} shows that if $(G, k, V_1, V_2, \ldots, V_k)$ is a yes-instance of \textsc{RMIS}, then $(G', p, k')$ is a yes-instance of \textsc{MHE}. We now give a proof in the other direction.
	
	Let $c$ be a coloring of $G'$ extending $p$ such that at least $k'$ edges are happy in $G'$ with respect to $c$. We may assume that $c$ is an optimal coloring of $G'$, i.e.\ it yields the maximum possible number of happy edges in $G'$. Then, by Claim \ref{claim:mhe_aux_ens}, $c(s_i) \in V_i$ for every $i \in [k]$. 
	
	Again, we argue that $S=\{c(s_1), c(s_2), \ldots, c(s_k)\}$ is a multicolored independent set in $G$. We start proving this fact with the following claim.
	
	\begin{claim}\label{claim:mhe_Ekr}
		There are at least $m+kr$ happy edges incident to the vertices of type $e_{uv}$ with respect to $c$.
	\end{claim}
	\begin{claimproof}
		We bound the number of happy edges not incident to the vertices of type $e_{uv}$.
		
		From $c(s_i)=v_i \in V_i$ follows that exactly $(3k + 2n) \cdot m$ edges are happy with respect to $c$ among edges that are incident to auxiliary vertices $a_{v,j}^k$. These are exactly the edges of types $1$ and $2$ in the proof of Claim \ref{claim:mhe_s_to_c}, and happy edges among them are counted in the same way as in the proof.
		
		The only other edges not incident to the vertices of type $e_{uv}$ are edges between $s_i$ and $t_{uv}^v$ for each $i \in [k]$, $v \in V_i$ and $uv \in E(G)$. Again, analysis of these edges is the same as the analysis of the edges of type $3$ in the proof of Claim \ref{claim:mhe_s_to_c}, and their number is at most $kr$.
		
		The only happy edges left are edges incident to $e_{uv}$ for some $uv \in E(G)$, hence the number of happy edges among them is at least $k'-(3k+2n)\cdot m-kr=m+kr$.
	\end{claimproof}

	The following claim, along with Claim \ref{claim:mhe_Ekr}, allows us to move from counting happy edges to counting happy vertices in $G'$.
	
	\begin{claim}\label{claim:mhe_edges_to_vertices}
		For any $uv \in E(G)$, $e_{uv}$ cannot be incident to more than two happy edges in $G'$ with respect to any coloring $c$ extending $p$. Moreover, if $e_{uv}$ is incident to exactly two happy edges, then either $t_{uv}^v$ is happy or $t_{uv}^u$ is happy with respect to $c$.
	\end{claim}
	\begin{claimproof}
		Take any $uv \in E(G)$. The neighbours of $e_{uv}$ are vertices $t_{uv}^u$ and $t_{uv}^v$, and also $s_i$ and $s_j$, where $u \in V_i$ and $v \in V_j$. In case $i=j$, $e_{uv}$ has only three neighbour vertices.
		
		We know that $c(t_{uv}^v)\neq c(t_{uv}^u)$, as $p(t_{uv}^v)\neq p(t_{uv}^u)$. Hence, only one of the edges between $(e_{uv}, t_{uv}^v)$ and $(e_{uv}, t_{uv}^u)$ can be happy with respect to the same coloring $c$.
		
		The same holds for $s_i$ and $s_j$: since $c(s_i)\neq c(s_j)$, only one of the edges between $e_{uv}$ and $s_i$ or between $e_{uv}$ and $s_j$ can be happy at the same time. In case $i=j$, there is the only edge $(e_{uv}, s_i)$ that can be either happy or not.
		
		Thus, happy edges incident to $e_{uv}$ can sum up to no more than two edges. Suppose now that $e_{uv}$ is incident to exactly two happy edges. Then these edges are $(e_{uv}, t_{uv}^x)$ and $(e_{uv}, s_{y})$ for some $x \in \{u,v\}$ and some $y \in \{i,j\}$. Hence, $c(s_y)=c(t_{uv}^x)=x$. By Claim \ref{claim:mhe_aux_ens}, $x \in V_y$, i.e.\ it is either $x=u$ and $y=i$ or $x=v$ and $y=j$. Hence, $t_{uv}^x$ and $s_y$ are connected by an edge in $G'$, and, since $t_{uv}^x$ has exactly two neighbours $e_{uv}$ and $s_y$ and they share the same color, $t_{uv}^x$ is happy with respect to $c$.
	\end{claimproof}

	By Claim \ref{claim:mhe_Ekr} and Claim \ref{claim:mhe_edges_to_vertices}, at least $kr$ vertices of type $e_{uv}$ are incident to exactly two happy edges with respect to $c$. Hence, there are at least $kr$ vertices of type $t_{uv}^v$ that are happy in $G'$ with respect to $c$. Note that these vertices remain happy with respect to $c$ even if we remove auxiliary vertices $a_{v,j}^k$ and edges between $e_{uv}$ and $s_i$, i.e.\ return to the original construction of $G'$ in the proof of Theorem \ref{thm:mhv_p3}.
	
	Thus, coloring $c$ yields at least $kr$ happy vertices of type $t_{uv}^v$ in the original construction of $G'$ in the proof of Theorem \ref{thm:mhv_p3}, hence we use Claim \ref{claim:mhv_constructing_is} to finish the proof of the first part of this theorem in the same way.
	
	We thereby have shown that \MHE~is $\W[1]$-hard when parameterized by the distance to graphs being a disjoint union of paths on three vertices. To prove the same for the distance to graphs being a disjoint union of cycles of length three, we note that in our construction of $G'$, endpoints of the paths are precolored vertices.
	
	Hence, we can add an edge between endpoints of each path, i.e.\ between $t_{uv}^v$ and $t_{uv}^u$ for each $uv \in E(G)$ and between $a_{v,j}^1$ and $a_{v,j}^3$ for each $v \in V(G)$ and $j \in [m]$, and just increase the parameter $k'$ by the number of newly-appeared happy edges. Namely, these are the edges between $a_{v,j}^1$ and $a_{v,j}^3$, thus we increase $k'$ by $n\cdot m$, and the other parts of the construction remain the same.
\end{proof}
\fi

\begin{corollary}\label{cor:mhe_w1_structural}
	\MHEfull~is $\W[1]$-hard with respect to parameters pathwidth, treewidth or clique-width, distance to cographs, feedback vertex set number.
\end{corollary}

The rest of the section focuses on the parameterized complexity of both \MHV~and \MHE~parameterized by the distance to cluster parameter.
We separate \MHE~and \MHV, showing that the former problem is $\W[1]$-hard with respect to this parameter, but the latter admits an \FPT-algorithm.
This answers an open question posed in works of Choudhari and Reddy \cite{Choudhari2018} and Misra and Reddy \cite{Misra2018}.

\begin{corollary}\label{cor:mhe_w1_cluster}
	\MHEfull~is $\W[1]$-hard when parameterized by the cluster vertex deletion number.
\end{corollary}
\begin{proof}
	Observe that graph consisting of  disjoint cycles of length three is a cluster graph. Then, by Theorem \ref{thm:mhe_p3_c3}, \textsc{MHE} is $\W[1]$-hard when parameterized by the distance to cluster graphs. 
\end{proof}

\begin{theorem}\label{thm:mhv_cluster_fpt}
	\MHVfull~can be solved in $\Ostar{(2d)^d}$ time, where $d$ is the distance to cluster parameter of the input graph.
\end{theorem}
\ifshort
\else
\begin{proof}
	We adapt algorithms of Misra and Reddy presented in \cite{Misra2018} in their proofs of \FPT\ membership result for both \MHV~and \MHE~parameterized by the vertex cover number and by the distance to clique parameters.
	
	Let $(G,p,k,S)$ be an instance of \MHV, and $S$ is a given minimum modulator to cluster of $G$.
	We describe an algorithm that works in $\Ostar{d^d}$, where $d=|S|$ is the distance to cluster parameter of $G$.
	Note that it is not necessary that $S$ is given explicitly.
	To find $S$, one can simply consider $G$ as an instance of \textsc{Cluster Vertex Deletion} parameterized by the solution size, and employ one of the algorithms working in time $\Ostar{c^d}$, let it be a simple $\Ostar{3^d}$ running time algorithm \cite{jansen1997disjoint}, or more sophisticated ones, working in $\Ostar{2^d}$ \cite{Hffner2008} or even in $\Ostar{1.9102^d}$ \cite{Boral2015} running time.
	Note that this would not change the overall $\Ostar{d^d}$ running time, since $c$ is a constant value.
	
	To solve the problem, the algorithm finds an optimal coloring of $(G,p)$.
	Let $c$ be an arbitrary optimal coloring of $(G,p)$.
	Firstly, the algorithm guesses what vertices of $S$ are happy with respect to $c$.
	Clearly, there are $2^d$ options to choose a subset $H \subseteq S$, and the algoithm considers each one of them.
	From now on, let $H$ be a fixed guess of the algorithm, i.e.\ it assumes that $H$ is the set of vertices of $S$ that are happy in $(G,p)$ with respect to $c$.
	
	At the other hand, $c$ partitions vertices of $S$ into groups of the same color, in other words, into equivalence classes.
	Obviously, such partitions can be enumerated in $\Ostar{d^d}$ time (if $\ell<d$, there are at most $\ell^d$ such partitions, that is even less).
	The algorithm guesses a partition corresponding to $c$.
	Let $S_1 \sqcup S_2 \sqcup \ldots S_t = S$ be a fixed guessed partition, where $t \le d$.
	Formally, a partition corresponding to $c$ should satisfy $c(u)=c(v) \Leftrightarrow \exists i: \; u, v \in S_i$ for each pair of vertices $u, v \in S$.
	For each $i$, the vertices in $S_i$ are assigned the same color, denote this color by $\alpha_i$.
	The actual value of the colors $\alpha_i$ is not known to the algorithm.
	Thus, $\alpha_1, \alpha_2, \ldots, \alpha_t$ are color variables and the algorithm is to determine what actual colors they should correspond to.
	Importantly, for distinct $i$ and $j$, $\alpha_i$ and $\alpha_j$ should correspond to distinct colors in $[\ell]$.
	
	For convenience, we introduce a partial function $\sigma: V(G) \nrightarrow \{\alpha_1, \alpha_2, \ldots, \alpha_t\}$ to the algorithm.
	If specified, a value $\sigma(v)$ denotes a color variable that corresponds to the color $c(v)$.
	Since distinct variables correspond to different colors, $\sigma$ can be viewed as a partial coloring of the vertices of $G$, just with the color variables used instead of actual colors.
	If both $\sigma(u)$ and $\sigma(v)$ are specified for a pair of vertices $u,v$, then $\sigma(u)=\sigma(v)$ if and only if $c(u)=c(v)$.
	In other words, $\sigma$ agrees with $c$.
	Since $c$ is a coloring of $(G,p)$, $\sigma$ agrees with $p$ as well.

	The purpose of $\sigma$ is to reflect restrictions on a coloring that are implied by the fixed guesses of the algorithm.
	That is, $\sigma$ should agree with the set of happy vertices $H$, and with the partition $S_1, S_2, \ldots, S_t$.
	Clearly, for each $i \in [t]$ and for each $v \in S_i$, $\sigma(v)=\alpha_i$.
	Also, since all vertices in $H$ are happy with respect to $c$, for each $v \in H$ and for each $u \in N(v)$, $\sigma(u)$ should equal $\sigma(v)$. The algorithm assigns values of $\sigma$ so that these restrictions are satisfied.
	If the fixed guesses correspond to an actual coloring, the function $\sigma$ satisfying these restrictions exists and is found easily by the algorithm.
	If $\sigma$ cannot be found, the algorithm stops working with the currently fixed guesses, since they do not correspond to any coloring of $G$.
	Note that the restrictions do not ensure that all vertices in $S\setminus H$ are unhappy.
	We formulate the main property of $\sigma$ in the following claim.
	
	\begin{claim}
		Let $c'$ be a coloring that agrees with $\sigma$ constructed by the algorithm.
		Then all vertices in $H$ are happy with respect to $c'$ in $G$ and $c'$ partitions the vertices of $S$ according to $S_1, S_2, \ldots, S_t$.
	\end{claim}

	Now the algorithm starts to find values of the color variables.
	This can be viewed as a constructing an injective function $\lambda: \{\alpha_1, \alpha_2, \ldots, \alpha_t\} \nrightarrow [\ell]$.
	Since $\sigma$ agrees with $p$, some values of $\lambda$ can be determined by the algorithm: if for some vertex $v$ both $\sigma(v)$ and $p(v)$ are specified, then $\lambda(\sigma(v))=p(v)$.
	The algorithm constructs $\lambda$ so that this property is satisfied.
	If it is impossible to construct an appropriate injective $\lambda$, the algorithm stops working with the current guesses and continues with another ones.
	
	When found, $\lambda$ allows to extend both $\sigma$ and $p$.
	If a correspondence between a color variable $\alpha_i$ and a color $a$ was established, i.e.\ $\lambda(\alpha_i)=a$, then we may assume that $\sigma(v)=\alpha_i \Leftrightarrow p(v)=a$.
	According to this, the algorithm extends $\sigma$ and $p$.
	Note that $p$ is no more an initial precoloring of $G$, since it was extended according to $\sigma$.
	
	We now want each vertex of $G$ to be assigned either a color (by $p$) or a color variable (by $\sigma$).
	If for a vertex $v \in V(G)$, neither $\sigma(v)$ nor $p(v)$ is assigned, we call $v$ \emph{unassigned}.
	Recall that all vertices in $S$ are assigned a color variable by $\sigma$.
	It is left to assign colors (by $p$) or color variables (by $\sigma$) to each of the vertices of the cluster, i.e.\ vertices in $V(G)\setminus S$.
	It turns out to be possible since we are looking for an optimal coloring that agrees with $p$ and $\sigma$.
	$p$ and $\sigma$ already ensure happiness of all vertices in $H$, so the algorithm can focus directly on happiness of the vertices in $V(G)\setminus S$.
	
	Consider a connected component in the cluster graph $G\setminus S$, say, a clique $C$.
	There are a few cases to consider.
	If $C$ contains two vertices that are assigned distinct colors by $p$ or distinct colors by $\sigma$, then all vertices in $C$ are unhappy with respect to any coloring that agrees with $p$ and $\sigma$.
	Thus, vertices in $C$ can be colored arbitrarily and there are no happy vertices among them.
	The algorithm assigns an arbitary color, say color $1$, to each unassigned vertex in $C$.
	Now $|p(C)|\le 1$ and $|\sigma(C)|\le 1$, consider easier case $p(C)=\emptyset$.
	In this case, the algorithm assigns color variables to unassigned vertices in $C$.
	$C$ can yield a happy vertex only if all vertices in $C$ receive the same color variable.
	Since $|\sigma(C)|\le 1$, there is an optimal coloring in which all vertices in $C$ are colored with the same color.
	If $|\sigma(C)|=1$, this is simply the color variable in $\sigma(C)$.
	Otherwise, the algorithm can simply determine how many happy vertices will $C$ yield if a color variable $\alpha_i$ is chosen.
	The only neighbours of vertices in $C$ outside of $C$ are vertices in $S$, and each vertex in $S$ is assigned a color variable by $\sigma$.
	Thus, it is easy to determine for a vertex whether it is happy if the whole clique is assigned a color variable $\alpha_i$.
	The algorithm chooses a color variable that gives the maximum possible number of happy vertices in $C$, and assigns it to each vertex in $C$.
	
	It is left to consider $p(C)=\{a\}$.
	In this case, $C$ can yield a happy vertex only if all vertices in $C$ receive color $a$.
	Hence, there is an optimal coloring where each unassigned vertex in $C$ is colored with color $a$.
	The algorithm assigns color $a$ to each unassigned vertex in $C$.
	It is left to determine how many happy vertices does $C$ contain.
	In contrast with the previous two cases, this depends on which color variable does correspond to color $a$.
	In case $\sigma(C)=\{\alpha_i\}$, $C$ can yield happy vertices only if $\alpha_i$ corresponds to $a$, and it is easy to find the number of happy vertices in $C$.
	In case $\sigma(C)=\emptyset$, each vertex in $C$ is colored with color $a$.
	Therefore, $C$ can contain some vertices that are happy in any case, that is, vertices that have no neighbours outside $C$.
	Since these vertices are always happy, the algorithm does not count them.
	Each other vertex in $C$ has at least one neighbour in $S$.
	If it has two neighbours with distinct color labels assigned, it can never be happy.
	Otherwise, all of its neighbours in $S$ are assigned the same color label, say $\alpha_j$, and the vertex is happy if and only if $\alpha_j$ corresponds to $a$.
	Thus, for each color variable $\alpha_j$ we get that $C$ yields a certain number of happy vertices if $\alpha_j$ corresponds to $a$.
	
	Summing up these values over all clique components, we get a weighted bipartite graph $B$.
	Left part of the graph corresponds to the color variables $\{\alpha_1, \alpha_2, \ldots, \alpha_t\}$, and the right one corresponds to the colors $[\ell]$.
	An edge between $\alpha_i$ in the left part and $a$ in the right part is assigned a weight equal to the number of happy vertices in $V(G)\setminus S$ in case $\alpha_i$ gets corresponding to $a$ (not counting vertices that are happy independently of this choice).
	Some color variables can already be assigned a color by $\lambda$, and the graph should reflect that.
	That is, for each $\alpha_i$ with $\lambda(\alpha_i)$ assigned, there is only one edge incident to $\alpha_i$ in $B$, and this edge is $\alpha_ia$.
	For each other color variable, there is each of $\ell$ possible edges presented in $B$.
	Clearly, a maximum-weight matching $M$ in $B$ that saturates all color variables yields an optimal way to assign colors to the color variables.
	
	The algorithm constructs graph $B$ and finds a maximum matching $M$ in $B$ in polynomial time.
	If $\alpha_i$ gets connected to $a$ in $M$, the algorithm extends $\lambda$ with $\lambda(\alpha_i)=a$.
	Since $G$ no more contains unassigned vertices, the optimal coloring can be simply constructed from the values of $p$, $\sigma$ and $\lambda$.
	The pseudo-code of the algorithm procedure working with a single pair of guesses is presented in Fig. \ref{fig:algo:cluster}.
\begin{figure}
	\centering
	
	\begin{algorithm}[H]
\LinesNumbered
		\TitleOfAlgo{$\texttt{find\_coloring}(G, p, S, H, \{S_1, S_2, \ldots, S_t\})$}
		\KwOut{An optimal coloring $c'$ of $(G,p)$ corresponding to the partition $\{S_1, S_2, \ldots, S_t\}$ such that all vertices in $H$ are happy with respect to $c'$; or Nothing, if $c'$ does not exist.}
		\BlankLine
		\SetAlgoVlined
		\DontPrintSemicolon
		
		initialize $\sigma: V(G) \nrightarrow \{\alpha_1, \alpha_2, \ldots, \alpha_t\}$ with no value assigned\;
		\lForEach{$i \in [t]$, $v \in S_i$}{
				$\sigma(v) \leftarrow \alpha_i$
		}
		
		\ForEach{$v \in H$, $u \in N(v)$}{
				$\sigma(u) \leftarrow \sigma(v)$, or \Return Nothing if $\sigma(u) \neq \sigma(v)$
}
		
		initialize $\lambda: \{\alpha_1, \alpha_2, \ldots, \alpha_t\} \nrightarrow [\ell]$ with no value assigned\;
		
		\ForEach{$v \in V(G)$ with both $\sigma(v), p(v)$ assigned}{
			$\lambda(\sigma(v)) \leftarrow p(v)$, or \Return Nothing if $\lambda(\sigma(v))\neq p(v)$\;
}
		
		\If{$\exists i, j \in [t]$ with $i\neq j$ but $\lambda(\alpha_i)= \lambda(\alpha_j)$}{
			\Return Nothing\;
		}
		
		\ForEach{$\alpha_i$ with $\lambda(\alpha_i)$ assigned}{
			\ForEach{$v \in V(G)$ with $p(v)=\lambda(\alpha_i)$ or $\sigma(v)=\alpha_i$}{
				$p(v) \leftarrow \lambda(\alpha_i)$;
				$\sigma(v) \leftarrow \alpha_i$\;
			}
		}
		
		initialize bipartite graph $B$ on $(\{\alpha_1, \alpha_2, \ldots, \alpha_t\}, [\ell])$ with zero-weight edges according to $\lambda$\;
		
		\ForEach{connected component $C$ in $G \setminus S$}{
			\eIf{$|p(C)|\ge 2$ or $|\sigma(C)| \ge 2$}{
				\ForEach{$v \in C$ with both $p(v),\sigma(v)$ unassigned}{
					$p(v) \leftarrow 1$\;
				}
			}{
				\ForEach{$i \in [t]$}{
					$C_i \leftarrow \{v \in C \mid \sigma(N[v])=\{\alpha_i\}\}$
				}
				\eIf{$p(C)=\emptyset$}{
					$i \leftarrow$ $\operatorname{argmax}_{i \in [t]} |C_i|$\;
					\ForEach{$v \in C$}{
						$\sigma(v) \leftarrow \alpha_i$
					}
				}{
					$a \leftarrow $ the color in $p(C)$\;
					\ForEach{$v \in C$ with both $p(v),\sigma(v)$ unassigned}{
						$p(v)\leftarrow a$\;
					}
					\ForEach{$i \in [t]$}{
						increase the weight of edge $\alpha_i a$ in $B$ by $|C_i|$ if it exists\;
					}
				}
			}
		}
		
		$M \leftarrow$ maximum-weight matching in $B$ saturating $\{\alpha_1, \alpha_2, \ldots, \alpha_t\}$\;
		\ForEach{$\alpha_i a \in M$}{
			$\lambda(\alpha_i) \leftarrow a$\; 
		}
		
		\ForEach{$v \in V(G)$ with $p(v)$ unassigned}{
			$p(v)\leftarrow \lambda(\sigma(v))$\;
		}
		\Return $p$\;

	\end{algorithm}
	
	\caption{A procedure finding an optimal coloring for fixed partitions of $S$.}
	\label{fig:algo:cluster}
\end{figure}
	The algorithm applies this procedure to each guess of the algorithm, and chooses the best among the resulting colorings.
	The correctness of the algorithm follows from the discussion.
	It is formulated in the following claim.
	
	\begin{claim}
		Let $c$ be a coloring of $(G,p)$, $H$ be the set of vertices in $S$ that are happy with respect to $c$, $\{S_1, S_2, \ldots, S_t\}$ be the partition of $S$ into groups of the same color according to $c$.
		Then $\texttt{find\_coloring}(G, p, S, H, \{S_1, S_2, \ldots S_t\})$ outputs a coloring $c'$ that yields at least the same number of happy vertices as $c$.
	\end{claim}

	Since the procedure works in polynomial time for any given partition, the overall running time is $2^d \cdot d^d \cdot n^{\O(1)}$.
	This finishes the proof.
\end{proof}
\fi

\section{Obtaining \W[2]-hardness}\label{sec:w2}
We are grateful for the anonymous reviewers of this paper for sharing ideas of how the statements of Theorem \ref{thm:mhv_p3} and Theorem \ref{thm:mhe_p3_c3} can be changed to obtain \W[2]-hardness with respect to structural parameters, strengthening corrollaries \ref{cor:mhv_w1_structural}, \ref{cor:mhe_w1_structural} and \ref{cor:mhe_w1_cluster}.
This section is dedicated to these \W[2]-hardness results.

\begin{theoremO}\label{thm:mhv_w2hard}
	\MHVfull~is $\W[2]$-hard when parameterized by the distance to graphs that are a disjoint union of stars.
\end{theoremO}

\begin{proofO}
	The proof is by reduction from the \textsc{Colourful Red-Blue Dominating Set} problem.
	\begin{problemx}
		\problemtitle{\textsc{Colourful Red-Blue Dominating Set (CRBDS)} \cite{Cygan2011}}
		\probleminput{A bipartite graph $G=(R \sqcup B, E)$, an integer $k$, and a coloring $c: R \to [k]$.}
		\problemparameter{$k$}
		\problemquestion{Does there exist a set $D \subseteq R$ of $k$ distinctly colored vertices such that $D$ is a dominating set of $B$?}
	\end{problemx}
	In \cite{Cygan2011}, Cygan et al.\ proved that \textsc{Colourful Red-Blue Dominating Set} is \W[2]-hard with respect to $k$.	
	Let $(G=(R\sqcup B, E),k,c)$ be an input of \textsc{CRBDS}.
	We assume that for each $v \in B$, $\deg_G(v)>1$.
	If $\deg_G(v)=0$ for some $v \in B$, then $(G,k,c)$ is a no-instance.
	If $\deg_G(v)=1$ for some $v \in B$, then the only neighbour of $v$ should be taken into the answer set $D$, and the instance can be trivially reduced.
	Analogously, we assume that for each color $i \in [k]$, there are at least two distinct vertices in $B$ that are colored with the color $i$ by $c$.
	
	We show how to construct an instance $(G',p',k')$ of \MHV~in polynomial time, such that $(G,k,c)$ is a yes-instance of \textsc{CRBDS} if and only if $(G',p',k')$ is a yes-instance of \MHV.
	Additionally, $G'$ is a graph such that at most $k$ vertices can be deleted from it to obtain a disjoint union of stars.
	
	Start from $G'$ being a graph consisting of no vertices and no edges.
	For each $v \in B$, introduce a new vertex $v$ in $G'$.
	Then, for each $u \in N(v)$, introduce a new vertex $u_v$ to $G'$ and connect it with the vertex $v$ by an edge.
	Note that for each vertex $u\in B$, exactly $|N(u)|$ vertices are introduced in $G'$, that are vertices $u_{v_1}, u_{v_2}, \ldots, u_{v_d}$, where $N(u)=\{v_1, v_2,\ldots, v_d\}$.
	Observe that $G'$ is now a graph consisting of $|B|$ connected components, and each of them is a star.
	
	The set of colors used in the precoloring $p'$ of $G'$ is identified with the set $R$.
	For each copy of a vertex $u \in R$, that is, for each $u \in R$ and for each $v \in N(u)$, precolor $u_v$ with a color $u$, i.e.\ put $p'(u_v)=u$.
	Thus, $|R|$ colors are used in the precoloring $p'$ of $G'$.
	
	Now introduce the selector gadget in $G'$, that consists of exactly $k$ vertices.
	For each $i \in [k]$, introduce new vertex $s_i$ in $G'$.
	Connect the vertex $s_i$ with all vertices of type $u_v$, such that $c(u)=i$, where $c$ is the coloring from the instance of \textsc{CRBDS}.
	That is, take each vertex that is colored with the color $i$ in the initial instance of \textsc{CRBDS} and connect $s_i$ with each copy of this vertex in $G'$.
	The construction of $G'$ is finished.
	Note that only the copies of the vertices in $R$ are precolored by $p'$ in $G'$.
	Finally, put $k'=|B|$.
	
	Analogously to the proof of Theorem \ref{thm:mhv_p3}, we now show that if $k'$ happy vertices in $(G',p')$ are achievable with some coloring $c'$ of $G'$ extending $p'$, then $D=\{c'(s_1), c'(s_2), \ldots,$ $ c'(s_k)\}$ is an answer to the initial instance $(G,k,c)$.
	And vice versa, if $D$ is a colorful dominating set of $B$, then it is enough to color vertices of the selector gadget in $G'$ correspondingly to the vertices of $D$.
	We first need the following claim.
	
	\begin{claim}
		There is an optimal coloring $c'$ of $(G',p')$ such that $c(c'(s_i))=i$ for each $i \in [k]$.
	\end{claim}
	\begin{claimproof}
		Take a coloring $c'$ and suppose that $c(c'(s_i))\neq i$ for some $i \in [k]$.
		Note that $s_i$ cannot be happy with respect to $c'$ as it has at least two distinctly precolored neighbours.
		Note that for each neighbor $u_v$ of $s_i$, $c(c'(u_v))=i$ by the construction of $G'$.
		Hence, $c'(s_i)\neq c'(u_v)$, so $u_v$ is not happy with respect to $c'$.
		Thus, one can change the color of $s_i$ in $c'$ to an arbitrary color without losing any happy vertices.
		
		Pick an arbitrary neighbor $u_v$ of $s_i$ and assign $c'(s_i)=c'(u_v)$, so that $c(c'(s_i))=i$ is now satisfied.
		Proceed with another $i$ until $c'$ satisfies the claim statement.
	\end{claimproof}

	This allows us to formulate the next claim.
	
	\begin{claim}\label{claim:mhv_w2_mhv_to_rbds}
		Let $c'$ be a coloring of $(G',p')$ such that at least $k'$ vertices are happy in $G'$ with respect to $c'$ and $c(c'(s_i))=i$ for each $i \in [k]$.
		Then $D=\{c'(s_1), c'(s_2), \ldots,$ $ c'(s_k)\}$ is an answer to $(G,k,c)$.
	\end{claim}
	\begin{claimproof}
		Note that the vertices of the selector gadget, i.e.\ vertices of type $s_i$ cannot be happy in $(G',p')$, as each of them has at least two neighbours with distinct colors in $p'$.
		For each $v \in B$, a copy of $v$ (one of $|B|$ star centers) cannot be happy in $(G',p')$ for the same reasons.
		Thus, the only vertices that can be happy in $(G',p')$ are precolored vertices.
		
		Note that in each of $|B|$ stars in $G'$, only one vertex can be happy in $(G',p')$ simultaneously, as all precolored vertices in $(G',p')$ are distinctly precolored, but a happy vertex should be colored by $c'$ with the same color as the star center.
		Since at least $k'=|B|$ vertices are happy in $G'$ with respect to $c'$, exactly one leaf vertex in each star is happy with respect to $c'$, and exactly $|B|$ vertices are happy in $G'$ with respect to $c'$ at all.
		
		That is, for each $v \in B$, there exists $u \in N(v)$, such that $c'(v)=c'(u_v)=u$.
		Note that $u_v$ is also connected to $s_{c(u)}$, and since $u_v$ is a happy vertex, $c'(s_{c(u)})=c'(u_v)$, hence $c'(v)=c'(s_{c(u)})$.
		This proves that $D$ is a dominating set of $B$, as for each vertex $v \in B$, $c'(v) \in N_G(v)$ and $c'(v) \in D$.
		The fact that $D$ consists of $k$ distinctly precolored vertices follows from the claim statement.
		Thus, $D$ is a colorful red-blue dominating set in $(G,k,c)$.
	\end{claimproof}

	The claim shows that if $(G',p',k')$ is a yes-instance, then $(G,k,c)$ is a yes-instance.
	We finally claim the other direction.
	
	\begin{claim}
		Let $D=\{u_1, u_2, \ldots, u_k\}$ be a dominating set of $B$ in $G$, where $u_i \in B$ and $c(u_i)=i$ for each $i \in [k]$.
		Extend the precoloring $p'$ to a full coloring $c'$ by putting $c'(s_i)=u_i$ for each $i\in[k]$ and, for each $v \in B$, put $c'(v)=u_j$, where $u_j \in D \cap N_G(v)$ is a vertex dominating $v$.
		Exactly $k'$ vertices are happy in $(G',p')$ with respect to $c'$.
	\end{claim}
	\begin{claimproof}
		To prove this claim, one can easily follow the construction of $c'$ and the proof of Claim \ref{claim:mhv_w2_mhv_to_rbds}.
	\end{claimproof}

	The last claim shows that if $(G,k,c)$ is a yes-instance, then $(G',p',k')$ is a yes-instance.
	The provided construction is polynomial.
	Deletion of the selector gadget makes $G'$ a disjoint union of star graphs, so the distance parameter is preserved.
	The proof is complete.
\end{proofO}

\begin{theoremO}
	\MHEfull~is $\W[2]$-hard when parameterized by the distance to graphs that are a disjoint union of stars and is $\W[2]$-hard when parameterized by the distance to graphs that are a disjoint union of cliques.
\end{theoremO}

\begin{proofO}
	As in the proof of Theorem \ref{thm:mhv_w2hard}, we again reduce from \W[2]-hard \textsc{Colourful Red-Blue Dominating Set}.
	Given an instance $(G,k,c)$ of \textsc{CRBDS}, we construct an instance $(G',p',k')$ in polynomial time.
	
	In the same way as in the proof of Theorem \ref{thm:mhv_w2hard}, start with $G'$ being an empty graph; for each $v \in B$ introduce a star in $G'$ with the center in $v$ and vertices $u_v$ for each $u \in N_G(v)$, precolored as $p'(u_v)=u$.
	Introduce the selector gadget vertices $s_1, s_2, \ldots, s_i$ to $G'$, but now connect them only to each of the star centers.
	Thus, the star centers and the vertices of the selector gadget induce a complete bipartite graph in $G'$.
	Note that the vertices of the selector gadget are not in any way connected to the precolored vertices of type $u_v$, as it was before in the proof of Theorem \ref{thm:mhv_w2hard}.
	To avoid a situation when it is profitable to color $s_i$ in a way that $c(c'(s_i))=i$, for each $u \in R$ introduce $|B|$ new vertices in $G'$, precolor them with color $u$ and connect them with $s_{c(u)}$.
	Finally, ask to make at least $k'=(2+k)\cdot|B|$ edges happy in $(G',p')$.
	The construction of $(G',p',k')$ is finished.
	Note that the deletion of the selector gadget from $G'$ still makes $G'$ a disjoint union of stars.
	
	We continue the proof with the series of claims similar to one in the proof of Theorem \ref{thm:mhv_w2hard}.
	
	\begin{claim}
		There is an optimal coloring $c'$ of $(G',p')$ such that $c(c'(s_i))=i$ for each $i \in [k]$. 
	\end{claim}
	\begin{claimproof}
		Take an optimal coloring $c'$ of $(G',p')$ and suppose that $c(c'(s_i))\neq i$ for some $i \in [k]$.
		For each $v \in B$, $s_i$ is connected to $v$ in $G'$.
		Note that any other neighbour of $s_i$ in $G'$ is precolored by $p'$ with a color $u$ such that $c(u)=i$.
		The color of $s_i$ in $c'$ does not satisfy this property, so only edges that are happy in $G'$ with respect to $c'$ and star centers.
		Hence, at most $|B|$ edges incident to $s_i$ are happy in $G'$ with respect to $c'$.
		
		Now take an arbitrary vertex $u \in B$ with $c(u)=i$ and change the color of $s_i$ in $c'$ to $u$.
		$s_i$ is now incident to at least $|B|$ happy edges.
		At most $|B|$ happy edges was lost with such operation, and at least $|B|$ happy edges was gained.
		Since $c'$ was an optimal coloring, the number of happy edges remained the same and $c'$ remains optimal.
		Continue this process until $c'$ satisfies the claim statement.
	\end{claimproof}

	\begin{claim}
		Let $c'$ be a coloring of $(G',p')$ such that at least $k'$ edges are happy in $G'$ with respect to $c'$ and $c(c'(s_i))=i$ for each $i \in [k]$.
		Then $D=\{c'(s_1), c'(s_2),$ $ \ldots,$ $ c'(s_k)\}$ is an answer to $(G,k,c)$.
	\end{claim}
	\begin{proof}
		Note that for each $v \in B$, a star centered in the copy of $v$ in $G'$ can contain at most one happy edge with respect to $c'$.
		Moreover, the vertex set of this star can be incident to at most two happy edges: one inside the star and one going from the star center to a vertex of the selector gadget.
		
		Edges that are not incident to the star components are edges between the selector gadget and the auxiliary vertices, that were introduced to ensure that $c(c'(s_i))=s_i$ for each $i \in [k]$.
		For each vertex $s_i$ of the selector gadget, there are exactly $|B|$ such happy edges incident to it, so there are exactly $k \cdot |B|$ such edges in total.
		Hence, at least $2|B|$ happy edges are incident to the star components.
		Since there are exactly $|B|$ stars and at most two happy edges can be incident to each one of them, exactly two happy edges are incident to each of the $|B|$ star centers.
		
		That is, for each $v \in B$, $c'(v)=c'(u_v)$ for some $u \in N_G(v)$, and $c'(v)=c'(s_i)$ for some $i \in [k]$.
		Similarly to the proof of Claim \ref{claim:mhv_w2_mhv_to_rbds}, this leads to that $D$ is an answer to $(G,k,c)$.
	\end{proof}

	The following claim again resembles a claim from the proof of Theorem \ref{thm:mhv_w2hard} and is clear.
	
	\begin{claim}
		Let $D=\{u_1, u_2, \ldots, u_k\}$ be a dominating set of $B$ in $G$, where $u_i \in B$ and $c(u_i)=i$ for each $i \in [k]$.
		Extend the precoloring $p'$ to a full coloring $c'$ by putting $c'(s_i)=u_i$ for each $i\in[k]$ and, for each $v \in B$, put $c'(v)=u_j$, where $u_j \in D \cap N_G(v)$ is a vertex dominating $v$.
		Exactly $k'$ edges are happy in $(G',p')$ with respect to $c'$.
	\end{claim}

	We resembled the proof of Theorem \ref{thm:mhv_w2hard} by proving a similar chain of claims for \MHE.
	Thus, the \W[2]-hardness of \MHE~with respect to the distance to a disjoint union of stars is proven.
	It is left to prove the same for the distance to cluster parameter.
	
	Note that after the deletion of the selector gadget $G'$ becomes a graph, where each component is a star where each leaf vertex is a precolored vertex.
	Moreover, all leaf vertices of each star are distinctly precolored.
	Thus, complementing each star to a clique in $G'$ yields edges that cannot be happy with respect to any coloring extending $p'$.
	Obtain a graph $G''$ by complementing each abovementioned star in $G'$ to a clique.
	Clearly $(G',p',k')$ and $(G'',p',k')$ are equivalent instances of \MHE, and $G''$ has the distance to cluster parameter being at most $k$.
	This finishes the proof.
\end{proofO}

\begin{corollaryO}
	\MHVfull~and \MHEfull~are both $\W[2]$-hard with respect to parameters pathwidth, treewidth or clique-width, distance to cographs, feedback vertex set number.
\end{corollaryO}
\begin{proofO}
	The proof is similar to the proof of Corollary \ref{cor:mhv_w1_structural}:
	star graphs are simultaneously cographs and forests, their pathwidth is equal to $1$.
\end{proofO}
	\section{\MHVfull~and \textsc{Node Multiway Cut}}\label{section:cut-relation}

This section reveals the connection between \MHVfull~and \textsc{Node Multiway Cut}.
This connection is a natural supplement of the straightforward connection of the edge versions of the problems, \MHEfull~and \textsc{Multiway Cut}.
It is more convenient for us to use a variation of \textsc{Node Multiway Cut}, called \textsc{Group Multiway Cut}, where terminal groups are used instead of singleton terminals. 

\begin{problemx}
	\problemtitle{\textsc{Group Multiway Cut} \cite{Chitnis2013}}
	\probleminput{A graph $G$ and pairwise disjoint sets of terminals $\{T_1, T_2, \ldots, T_\ell\}$, and an integer $k$.}
	\problemquestion{Is there a set $S \subseteq V(G)$ of size at most $k$ such that $G\setminus S$ has no $u-v$ path for any $u \in T_i, v \in T_j$ and $i\neq j$?}
\end{problemx}

We start with the following crucial lemma.

\begin{lemma}\label{lemma:mhv_paths}
	Let $G$ be a graph with precoloring $p$. Let $H \subseteq \mathcal{H}(G,p)$ be an arbitrary subset of its potentially happy vertices. Then a coloring $c$ extending $p$, so that all vertices in $H$ are happy with respect to $c$, exists if and only if there exists no path $u_1, u_2, \ldots, u_t$ in $G$, such that $u_1$ and $u_t$ are precolored and $p(u_1)\neq p(u_t)$, and for each $i \in [t-1]$, either $u_i$ or $u_{i+1}$ is in $H$.
\end{lemma}
\ifshort
\else
\begin{proof}
	Let $(G,p)$ be a graph with precoloring, and $H \subseteq \mathcal{H}(G,p)$ be a subset of its potentially happy vertices. We prove the statement first in the direction that the existence of the coloring implies the non-existence of the path.
	
	Suppose there is a coloring $c$ of $(G,p)$ such that all vertices in $H$ are happy with respect to $c$. Suppose by contradiction that there exists a path $u_1, u_2, \ldots, u_t$ in $(G,p)$, such that $p(u_1)\neq p(u_t)$ and for each $i \in [t-1]$, either $u_i$ is in $H$ or $u_{i+1}$ is in $H$. All vertices in $H$ are happy with respect to $c$, hence all edges incident to the vertices in $H$ are happy with respect to $c$. Thus, for each $i \in [t-1]$, the edge $(u_i, u_{i+1})$ is happy with respect to $c$, i.e.\ $c(u_i)=c(u_{i+1})$. This contradicts $p(u_1)\neq p(u_t)$.
	
	Let us prove in the other direction. Suppose that there exists no path $u_1, u_2,$ $ \ldots,$ $ u_t$ satisfying the conditions in $(G,p)$. Construct a coloring $c$ as follows. For each vertex $v$ in $V(G)$, such that $v$ is neither in $H$ nor a neighbour of a vertex in $H$, remove $v$ from $G$ and set $c(v)=p(v)$ if $v$ is precolored, or set $c(v)$ to an arbitrary color otherwise. No deleted vertex influences a happiness of a vertex in $H$, so it is left to color the vertices of the remaining graph $G[N[H]]$ so that each vertex in $H$ is happy. Consider a connected component $C$ of $G[N[H]]$ and observe that there are no two vertices in $C$ that are precolored with a different color. Suppose it's not true, then there exists a path between some vertices $u$ and $w$, $p(u)\neq p(w)$, in $(G[N[H]],\left.p\right|_{N[H]})$. Note that each edge of this path has an endpoint in $H$, and obtain a contradiction.
	
	Thus, each connected component of $(G[N[H]],\left.p\right|_{N[H]})$ either contains no precolored vertices, or all precolored vertices in this component are of the same color. That is, each connected component of $(G[N[H]],\left.p\right|_{N[H]})$ can be colored with a single color, so all vertices of $G[N[H]]$ are happy. Hence, all vertices in $H$ are also happy with respect to the same coloring in $(G, p)$. This concludes the proof.
\end{proof}
\fi

\begin{theorem}\label{thm:mhv_gmc_connection}
	Let $(G, p, k)$ be an instance of \MHVfull. Then $(G,p,k)$ is a yes-instance of \MHVfull~if and only if\linebreak $(G^2[\mathcal{H}(G,p)], \{\mathcal{H}_1(G,p), \mathcal{H}_2(G,p), \ldots, \mathcal{H}_\ell(G,p)\}, |\mathcal{H}(G, p)|-k)$ is a yes-instance of \textsc{Group Multiway Cut}.
\end{theorem}
\begin{proof}
	Let $(G,p,k)$ be a yes-instance of \MHV. We show that $(G^2[\mathcal{H}(G,p)],$ $ \{\mathcal{H}_1(G,p),\mathcal{H}_2(G,p),\ldots,\mathcal{H}_\ell(G,p)\},$ $ |\mathcal{H}(G, p)|-k)$ is a yes-instance of \textsc{Group Multiway Cut}. Since $(G,p,k)$ is a yes-instance, there is a coloring $c$ such that at least $k$ vertices of $(G,p)$ are happy with respect to $c$. Let $H$ be a set of any $k$ of these vertices, i.e.\ $|H|=k$ and all vertices in $H$ are happy with respect to $c$ in $(G,p)$. 
	
	Observe that any path in $G$ whose all edges are incident to at least one vertex of $H$, corresponds to a simple path in $G^2[H]$. 
	Indeed, let $u_1, u_2, \ldots, u_t$ be a path in $G$ such that $u_i \in H$ or $u_{i+1} \in H$ for each $i \in [t-1]$. Let $u_1^H, u_2^H, \ldots, u_{t_1}^H$ be the subsequence of $u_1, \ldots, u_t$ of vertices in $H$ ($H\cap \{u_1, \ldots, u_t\}=\{u_1^H, \ldots, u_{t_1}^H\}$). Note that for each $i \in [t_1-1]$, $u_{i}^H$ and $u_{i+1}^H$ are either consequent in $u_1, \ldots, u_t$ or there is only one vertex between them in $u_1, \ldots, u_t$. That is, there is an edge between $u_i^H$ and $u_{i+1}^H$ in $G^2[H]$. Thus, $u_1^H, u_2^H, \ldots, u_{t_1}^H$ is a path in $G^2[H]$. Vice versa, any simple path in $G^2[H]$ corresponds to paths in $G$ which edges are incident to vertices in $H$.
	
	Since all vertices $H$ are happy in $(G,p)$, by Lemma \ref{lemma:mhv_paths}, there is no path between differently precolored vertices with all edges incident to at least one vertex in $H$. Consider $G^2[H]$ and suppose that there exists a path between vertices $v$ and $w$ in $G^2[H]$, such that $v \in \mathcal{H}_i(G,p)$ and $w \in \mathcal{H}_j(G,p)$, and $i \neq j$. As shown above, this path corresponds to a path between $v$ and $w$ in $G$, and all edges of this path are incident to $H$. Since $v \in \mathcal{H}_i(G,p)$, there is a precolored vertex $v' \in N[v]$ with $p(v')=i$. Similarly, there is a $w' \in N[w]$ with $p(w')=j$. There is a path between $v'$ and $w'$ in $G$ with all edges incident to $H$ and $p(v')\neq p(w')$, a contradiction. Hence, no vertices in different sets of terminals $\mathcal{H}_i(G,p)$ and $\mathcal{H}_j(G,p)$ are connected in $G^2[H]$. Thus, $\mathcal{H}(G, p)\setminus H$ is an answer to $(G^2[\mathcal{H}(G,p)], \{\mathcal{H}_i(G,p)\},$ $ |\mathcal{H}(G, p)|-k)$, so it is a yes-instance of \textsc{Group Multiway Cut}.
	
	The proof in the other direction is similar: if $S$, $(|S|=|\mathcal{H}(G,p)|-k)$, is a solution to the instance of \textsc{Group Multiway Cut}, then all $k$ vertices in $\mathcal{H}(G,p)\setminus S$ can be happy simultaneously in $(G,p)$.
\end{proof}

Theorem \ref{thm:mhv_gmc_connection} shows the importance of potentially happy vertices in the input of \MHV.
Other vertices are playing role of common neighbours or precolored neighbours for potentially happy vertices.
Note that the sets $\mathcal{H}(G,p)$ and $\mathcal{H}_i(G,p)$ are computable in polynomial time.
Thus, an instance of \MHV~can be compressed in order to contain only useful information about potentially happy vertices.
We formulate this in the following corollary.

\begin{corollaryO}\label{cor:mhv_kernel_pot_happy}
	\MHVfull, parameterized by the number of potentially happy vertices $h$, (i) admits a  polynomial compression into \textsc{Group Multiway Cut} with $h$ vertices and (ii) admits a kernel with $\O(h^2)$ vertices and edges.
\end{corollaryO}
\begin{proofO}
	(i) is a direct corollary of Theorem \ref{thm:mhv_gmc_connection}.
	To prove (ii) and obtain a kernel with $\O(h^2)$ vertices and edges for an instance $(G,p,k)$ of \MHV, compress it firstly to an equivalent instance $(G^2[\mathcal{H}(G,p)], \{\mathcal{H}_i(G,p)\}, h-k)$ of \textsc{Group Multiway Cut}. Then, transform this instance back to an equivalent instance $(G', p', k)$ of \MHV~as follows. Construct $G'$ as a subdivision of $G^2[\mathcal{H}(G,p)]$, and then introduce two vertices $t_1$ and $t_2$ to $G'$, connect them by an edge $(t_1, t_2)$, and connect them both to each vertex that was introduced to $G'$ because of the subdivision of an edge. Finally, for each $i \in [\ell]$ and each $v \in \mathcal{H}_i(G,p)$, set $p'(v)=i$; then set $p'(t_1)=1$ and $p'(t_2)=2$. Observe that $\mathcal{H}(G',p')=\mathcal{H}(G,p)$ and $\mathcal{H}_i(G',p')=\mathcal{H}_i(G,p)$, as no newly-introduced vertex is potentially happy, and no potentially happy vertex in $G'$ is adjacent to a vertex precolored by $p'$. Moreover, $(G')^2[\mathcal{H}(G',p')]=G^2[\mathcal{H}(G,p)]$ by means of the subdivision. Hence, $(G',p',k)$ is an instance equivalent to $(G,p,k)$, and there is $\O(|\mathcal{H}(G,p)|^2)$ vertices and edges in $G'$.
\end{proofO}

Another interesting consequence of Theorem \ref{thm:mhv_gmc_connection}, along with Corollary \ref{cor:mhv_w1_structural}, is a lower bound on algorithms for \textsc{Group Multiway Cut} parameterized by clique-width.

\begin{corollary}
	\textsc{Group Multiway Cut} is $\W[1]$-hard when parameterized by the clique-width of the input graph.
\end{corollary}
\begin{proof}
	By Corollary \ref{cor:mhv_w1_structural}, \MHVfull~is \W[1]-hard when parameterized by the clique-width of the input graph. Take an instance $(G,p,k)$ of \MHV. As shown by Todinca in \cite{Todinca2003}, if $G$ has clique-width $t$, then the power $G^c$ of $G$ has clique-width at most $2tc^t$. Hence, $G^2$ has clique-width at most $2t2^t$. Then, as shown by Courcelle and Olariu in \cite{Courcelle2000}, every induced subgraph of a graph of clique-width $t$ has clique-width at most $t$, so $G^2[\mathcal{H}(G,p)]$ has clique-width at most $2t2^t$ as well. So, in an instance $(G^2[\mathcal{H}(G,p)], \{\mathcal{H}_i(G,p)\}, |\mathcal{H}(G,p)|-k)$ of \textsc{Group Multiway Cut} equivalent to the instance $(G,p,k)$, the clique-width of the input graph is bounded if the clique-width of $G$ is bounded. Since the reduction from \MHV~to \textsc{Group Multiway Cut} is polynomial, the corollary statement follows.
\end{proof}

In contrast, we have that \textsc{Node Multiway Cut} is in \FPT~when parameterized by the clique-width of the input graph.
We present an algorithm solving \textsc{Node Multiway Cut} using dynamic programming on a $w$-expression of $G$.

\begin{theoremO}\label{thm:nmc_cw}
	\textsc{Node Multiway Cut} can be solved in $(w+3)^{2w} \cdot n^{\O(1)}$, if a $w$-expression of $G$ is given.
\end{theoremO}
\begin{proofO}
	Let $(G,T,k,\Psi)$ be an instance of \textsc{Node Multiway Cut} with a given $w$-expression $\Psi$ of $G$.
	To solve the instance, we employ dynamic programming.
	To help undestanding it, we suggest to consider the following.
	Consider a subexpression $\Phi$ of $\Psi$.
	Let $S$ be a vertex subset of  $V(G_\Phi)\setminus T$ whose deletion from $G_\Phi$ disconnects all terminals in $T$ from each other.
	That is, $S$ is a possible answer for the smaller graph $G_\Phi$.
	Let $V_i(G_\Phi \setminus S)$ denote the set of vertices with label $i$ in $G_\Phi \setminus S$.
	Let $R_j(G_\Phi \setminus S)$ denote the set of vertices reachable from the terminal $t_j$ in $G_\Phi \setminus S$ (including $t_j$ itself).
	If $t_j \notin V(G_\Phi)$, then let $R_j(G_\Phi \setminus S)=\emptyset$.
	Since $S$ is an answer to the instance, all $R_j(G_\Phi \setminus S)$ are disjoint.
	
	Consider an arbitrary vertex label $i \in [w]$ in $\Phi$.
	We distinguish four types of labels depending on what terminal vertices in $T$ the vertices in $V_i(G_\Phi \setminus S)$ are connected to.
	If there is no vertex with label $i$ in $G_\Phi \setminus S$, i.e.\
	$V_i(G_\Phi\setminus S)=\emptyset$, we say that $i$ is of \emph{$\emptyset$-type} in $G_\Phi \setminus S$.
	If all vertices with label $i$ in $G_\Phi \setminus S$ are not connected with any of the terminals in $T$, i.e.\ $V_i(G_\Phi \setminus S) \cap R_j(G_\Phi \setminus S)=\emptyset$ for each $j \in [\ell]$, and $V_i(G_\Phi \setminus S)$ is not empty, we say that $i$ is of \emph{$0$-type} in $G_\Phi \setminus S$.
	If a label $i$ is not of these two types, then it must contain a vertex that is connected to some terminal in $T$.
	If there are two vertices with label $i$, each of these two is connected to a terminal in $G_\Phi \setminus S$, but these two terminals are distinct, i.e.\ $V_i(G_\Phi\setminus S)\cap R_j(G_\Phi \setminus S)\neq \emptyset$ for at least two values of $j$, we say that $i$ is of \emph{$2$-type}.
	Otherwise, there is exactly one terminal, say $s_j$, reachable from $V_i(G_\Phi \setminus S)$ in $G_\Phi \setminus S$, and we say that $i$ is of \emph{$1$-type}.
	Note that in that case $s_j$ is necessarily a vertex of $G_\Phi$.
	
	In this way, the set of labels $[w]$ become partitioned into four sets $P_\emptyset$, $P_0$, $P_1$ and $P_2$, depending on the value of $\Phi$ and $S$.
	Note that each $1$-type label corresponds to a single teminal.
	Thus, the set $P_1$ can be further partitioned into a family of sets $\mathcal{P}_1=\{P_{1,1}, P_{1,2},\ldots, P_{1,p}\}$, $p \le w$, where each $P_{1,j}$ is a set containing all $1$-type labels corresponding to a certain terminal.
	Note that despite there are $\ell$ terminals in $T$ and $\ell$ may be much greater than $w$, $\mathcal{P}_1$ always contains no more than $w$ sets, since it is a partition of $P_1$.
	We obtain a partition $(P_\emptyset, P_0, \mathcal{P}_1, P_2)$ according to $S$ and $\Phi$.
	Thus, for fixed $\Phi$ and $S$, there is exactly one partition corresponding to $S$.
	Note that different values of $S$ may lead to the same partitions in $\Phi$.
	Also, for some partitions of labels of $\Phi$, there may be no corresponding values of $S$.
	We now bound the number of possible partitions.
	
	\begin{claim}
		There are at most $(w+3)^w$ possible partitions $(P_\emptyset, P_0, \mathcal{P}_1, P_2)$ of the label set $[w]$, and all possible partitions can be enumerated in $\Ostar{(w+3)^w}$ time.
	\end{claim}
	\begin{claimproof}
		Observe that any partition contains at most $w+3$ sets.
		Since each element in $[w]$ appears in exactly one set of the partition, there are at most $(w+3)^w$ possible partitions and they can be enumerated easily.
	\end{claimproof}
	
	We are ready to introduce the dynamic programming. For each subexpression $\Phi$ of $\Psi$, and each possible partition $(P_\emptyset, P_0, \mathcal{P}_1, P_2)$, $$OPT(\Phi, P_\emptyset, P_0, \mathcal{P}_1, P_2)$$ stores the minimum size of $S$ among all possible solutions $S$ for $G_\Phi$ (that is, $S \subseteq V(G_\Phi) \setminus T$ and deletion of $S$ disconnects all terminals in $G_\Phi$ from each other), such that $S$ corresponds to the partition $(P_\emptyset, P_0, \mathcal{P}_1, P_2)$.
	If there is no solution $S$ corresponding to the partition, $OPT(\Phi, P_\emptyset, P_0, \mathcal{P}_1, P_2)=\infty$.
	Clearly, $OPT$ consists of at most $|\Psi| \cdot (w+3)^w$ states.
	The initial instance $(G, T, k, \Psi)$ is a yes-instance if and only if $$\min_{(P_\emptyset, P_0, \mathcal{P}_1, P_2)} OPT(\Psi, P_\emptyset, P_0, \mathcal{P}_1, P_2) \le k.$$
	
	We now show how we compute the values of $OPT$.
	We compute the values of $OPT$ going from smaller subexpressions of $\Psi$ to larger.
	Thus, if we are to compute the values of $OPT$ for some subexpression $\Phi$ of $\Psi$, we have all values of $OPT$ calculated for all subexpressions of $\Phi$.
	Let $\Phi$ be a fixed subexpression of $\Psi$.
	To simplify our task, we do not consider a fixed stage $OPT(\Phi, P_\emptyset, P_0, \mathcal{P}_1, P_2)$ of $OPT$ and compute its value at once.
	Instead, we initialize all values of $OPT(\Phi, \cdot)$ with $\infty$, and then update them considering values of $OPT$ for smaller subexpressions.
	This procedure depends on the topmost operator in $\Phi$.
	
	\begin{enumerate}
		\item $\Phi=i(v)$.
		That is, $G_\Phi$ consists of a single vertex with label $i$.
		Each label, except for $i$, is of $\emptyset$-type.
		One may either choose to pick $v$ in the solution and delete it, or  not to delete it.
		In other words, the possible values of $S$ are $\emptyset$ and $v$.
		However, if $v \in T$, we should always leave the vertex in the graph.
		
		If we choose not to delete $v$, $i$ stays a label of $1$-type. Hence, we put $OPT(\Phi,$ $ [w]\setminus \{i\}, \emptyset, \{\{i\}\}, \emptyset)=0$.
		If we choose to delete $v$ (in that case, $v \notin T$), $G_\Phi$ becomes empty, so we put $OPT(\Phi, [w], \emptyset, \emptyset, \emptyset)=1$.
		Clearly, all other values $OPT(\Phi,\cdot)$ should be left equal $\infty$.
		
		\item $\Phi=\rho_{i \to j} \Phi'$.
		Note that any solution $S$ for $G_{\Phi'}$ is a solution for $G_\Phi$, and vice versa.
		Let $S$ be a solution that corresponds to a partition $(P'_\emptyset, P'_0, \mathcal{P}'_1, P'_2)$ in $\Phi'$.
		$S$ also corresponds to some partition $(P_\emptyset, P_0, \mathcal{P}_1, P_2)$ in $\Phi$, and it occurs that this partition is easy to find if the partition $(P'_\emptyset, P'_0, \mathcal{P}'_1, P'_2)$ is given.
		Since $\Phi$ differs from $\Phi'$ only in renaming label $i$ to label $j$, $(P_\emptyset, P_0, \mathcal{P}_1, P_2)$ differs from $(P'_\emptyset, P'_0, \mathcal{P}'_1, P'_2)$ only in positions of labels $i$ and $j$.
		In fact, $V_i(G_\Phi)=\emptyset$, so $i$ is of $\emptyset$-type in $G_\Phi \setminus S$.
		Thus, $i \in P_\emptyset$.
		It is left to determine the position of label $j$ in the partition $(P_\emptyset, P_0, \mathcal{P}_1, P_2)$.
		
		Since $V_j(G_\Phi)=V_i(G_{\Phi'}) \sqcup V_j(G_{\Phi'})$, the type of $j$ in $G_\Phi$ depends on types of $i$ and $j$ in $G_{\Phi'}$.
		If at least one of $i$ or $j$ are of $\emptyset$-type in $G_{\Phi'}\setminus S$, without loss of generality let it be $i$, then $V_j(G_\Phi)=V_j(G_{\Phi'})$.
		Hence, the type of $j$ in $G_\Phi$ is equal to the type of $j$ in $G_\{\Phi'\}\setminus S$, and its position in the partition (if it is $1$-type) remains the same.
		When neither of $i$ and $j$ are of $\emptyset$-type, but at least one of them is of $0$-type, again, let it be $i$, the type and position of $j$ in $G_\Phi$ remains the same as in $G_{\Phi'}$.
		If at least one of $i$ or $j$ are of $2$-type in $G_{\Phi'}\setminus S$, then, clearly, $j$ is of $2$-type in $G_{\Phi}\setminus S$.
		
		It is left to consider the case when both $i$ and $j$ are of $1$-type in $G_{\Phi'}\setminus S$, i.e.\ $i,j\in \mathcal{P}'_1$.
		Consider $i$ and $j$ belonging to distinct sets in $\mathcal{P}'_1$, that is, there is a terminal reachable from $V_i(G_{\Phi'})\setminus S$ and a terminal reachable from $V_j(G_{\Phi'}) \setminus S$, and these terminals are distinct.
		Then, both these terminals are reachable from $V_j(G_{\Phi})\setminus S$, so $j$ should be of $2$-type in $G_{\Phi}\setminus S$, i.e.\ $j \in P_2$.
		In the only case left, $i$ and $j$ belong to the same set in $\mathcal{P}'_1$.
		Clearly, the position of $j$ should not change in that case.
		
		Thus, to compute all values $OPT(\Phi, \cdot)$, we iterate over all possible partitions $(P'_\emptyset, P'_0, \mathcal{P}'_1, P'_2)$ of labels in $G_{\Phi'}$.
		For a fixed partition, we find the corresponding partition $(P_\emptyset, P_0, \mathcal{P}_1, P_2)$ of labels in $G_{\Phi}$ in polynomial time as described above.
		Finally, we update $OPT(\Phi,P_\emptyset, P_0, \mathcal{P}_1, P_2)$ with the value of $OPT(\Phi',P'_\emptyset, P'_0, \mathcal{P}'_1, P'_2)$.
		
		\item $\Phi=\Phi' \oplus \Phi''$.
		That is, $G_\Phi$ is a disjoint union of $G_{\Phi'}$ and $G_{\Phi''}$.
		Thus, if $S$ is a solution for $G_\Phi$, then $S'=S\cap V(G_{\Phi'})$ is a solution for $G_{\Phi'}$, and $S''=S\cap V(G_{\Phi''})$ is a solution for $G_{\Phi''}$, and vice versa, two solutions $S'$ and $S''$ give a solution $S'\sqcup S''$ for $G_\Phi$.
		Thus, any solution for $G_\Phi$ is a union of solutions for $G_{\Phi'}$ and $G_{\Phi''}$.
		Again, we consider all possible partitions $(P'_\emptyset, P'_0, \mathcal{P}'_1, P'_2)$ and $(P''_\emptyset, P''_0, \mathcal{P}''_1, P''_2)$ for solutions $S'$ and $S''$ respectively in $G_{\Phi'}$ and $G_{\Phi''}$, and show how to find a partition $(P_\emptyset, P_0, \mathcal{P}_1, P_2)$ corresponding to their union $S=S'\sqcup S''$ in $G_\Phi$.
		
		Since $V_i(G_\Phi)\setminus S=(V_i(G_{\Phi'}) \setminus S') \sqcup (V_i(G_{\Phi''})\setminus S'')$ for each label $i$, it is enough to show how to determine position of label $i$ in the partition for $G_{\Phi}$ knowing its position in the partitions for $G_{\Phi'}$ and $G_{\Phi''}$.
		Consider the types of $i$ in $G_{\Phi'}\setminus S'$ and $G_{\Phi''}\setminus S''$.
		Cases when at least one of these types is not $1$-type are handled in the same way as above for the relabelling operator.
		Consider the case when $i$ is of $1$-type both in $G_{\Phi'}\setminus S'$ and in $G_{\Phi''}\setminus S''$.
		That means that there is a terminal in $G_{\Phi'}$ reachable from the vertices with label $i$ in $G_{\Phi'}\setminus S'$, and a terminal in $G_{\Phi''}$ reachable from the vertices with label $i$ in $G_{\Phi''}\setminus S''$.
		Since $V(G_{\Phi'})$ and $V_(G_{\Phi''})$ are disjoint, these two terminals are distinct.
		Hence, in $G_\Phi\setminus S$ there are at least two terminals reachable from vertices with label $i$.
		Therefore, $i$ should receive $2$-type in $G_\Phi \setminus S$. 
		
		Again, to compute values of $OPT(\Phi,\cdot)$, we iterate over all possible partitions $(P'_\emptyset, P'_0, \mathcal{P}'_1, P'_2)$ and $(P''_\emptyset, P''_0, \mathcal{P}''_1, P''_2)$.
		Having these two partitions fixed, we find the combined partition $(P_\emptyset, P_0, \mathcal{P}_1, P_2)$ and update $OPT(\Phi, P_\emptyset, P_0, \mathcal{P}_1,$ $ P_2)$ with $OPT(\Phi', P'_\emptyset, P'_0, \mathcal{P}'_1, P'_2)+OPT(\Phi'', P''_\emptyset, P''_0, \mathcal{P}'_1, P''_2)$.
		
		\item $\Phi=\eta_{i,j}\Phi'$.
		That is, $G_\Phi$ is a graph obtained by introducing all possible edges with endpoints having labels $i$ and $j$.
		It is again easy to see that if $S$ is a solution for $G_\Phi$, then $S$ is a solution for $G_{\Phi'}$.
		However, a solution for $G_{\Phi'}$ does not yield a solution $G_{\Phi}$.
		It occurs again that to check that a solution $S$ for $G_{\Phi'}$ is a suitable solution for $G_{\Phi}$, it suffices to know the partition of labels for $S$ and $\Phi'$.
		
		Let $(P'_\emptyset, P'_0, \mathcal{P}'_1, P'_2)$ be a partition corresponding to $S$ in $G_{\Phi'}$.
		We show how to check that $S$ is a suitable solution for $G_\Phi$, and find a partition $(P_\emptyset, P_0, \mathcal{P}_1, P_2)$ corresponding to $S$ in $G_\Phi$.
		Clearly, it is enough to consider the types of labels $i$ and $j$ in $G_{\Phi'}\setminus S$.
		If at least one of $i$ and $j$ are of $\emptyset$-type, then no edge is actually added by the operator and $G_{\Phi}=G_{\Phi'}$.
		Thus, $S$ is a suitable solution for $G_{\Phi}$ and the partition for $G_{\Phi}\setminus S$ remains the same as for $G_{\Phi'}\setminus S$.
		If neither of $i$ and $j$ are of $\emptyset$-type, but at least one of them is of $2$-type in $G_{\Phi'}\setminus S$, then $S$ is not a suitable solution for $G_\Phi$.
		Indeed, without loss of generality let $i$ be of $2$-type, i.e.\ vertices with label $i$ are connected to two distinct terminals.
		Since there is at least one vertex with label $j$, adding all edges between vertices with labels $i$ and $j$ connects these two terminals together.
		
		It is left to consider cases when $i$ and $j$ are of $0$-type or of $1$-type in $G_{\Phi'}\setminus S$.
		If both of them are of $0$-type, then adding edges between vertices of these two labels does not yield any connection between terminals.
		Hence, $S$ is a suitable solution for $G_\Phi$.
		Also, $i$ and $j$ remain being of $0$-type in $G_\Phi\setminus S$, so the partition remains the same.
		If one label is of $0$-type, and the other is of $1$-type, then adding edges does not yield any connection between two distinct terminals, and $S$ is a suitable solution for $G_\Phi$.
		However, since vertices with both labels are now connected to the same terminal in $G_\Phi\setminus S$, the partition should be changed: both labels should receive $1$-type and go into the same set in $\mathcal{P}_1$.
		It is left to consider the case when both $i$ and $j$ are of $1$-type in $G_{\Phi'}$.
		If $i$ and $j$ are in the same set in $\mathcal{P}'_1$, then no two terminals become connected, so $S$ is a suitable solution for $G_\Phi$.
		The partition remains the same.
		Otherwise, $i$ and $j$ are in distinct sets inside $\mathcal{P}_1$, i.e.\ the vertices with labels $i$ and $j$ are connected to two distinct terminals.
		Clearly, in that case, $S$ is not a suitable solution for $G_\Phi$.
		
		We have shown how to obtain the partition of labels $(P_\emptyset, P_0, \mathcal{P}_1, P_2)$ for $G_\Phi\setminus S$ knowing the partition $(P'_\emptyset, P'_0, \mathcal{P}'_1, P'_2)$ for $G_{\Phi'} \setminus S'$, and to ensure
		that any solution for $G_{\Phi'}$ corresponding to $(P'_\emptyset, P'_0, \mathcal{P}'_1, P'_2)$ is a suitable solution for $G_\Phi$.
		Thus, we again iterate over all possible partitions for $(P'_\emptyset, P'_0, \mathcal{P}'_1, P'_2)$ and check that it corresponds to a suitable solution of $G_\Phi$.
		If it does, we find the partition for $G_\Phi \setminus S$ and update $OPT(\Phi, P_\emptyset, P_0, \mathcal{P}_1, P_2)$ with $OPT(\Phi', P'_\emptyset, P'_0,$ $ \mathcal{P}'_1, P'_2)$.	
	\end{enumerate}
	
	We have shown how to compute the values of $OPT$.
	Each computation step occurs for a fixed subexpression $\Phi$, and is done by considering all possible partitions for child subexpressions of $\Phi$.
	Since there are always at most two child subexpression (in the case of the disjoint union operator), so the heaviest computation step takes $(w+3)^{2w} \cdot n^{\O(1)}$ running time.
	Having all values of $OPT$ calculated, the answer to the initial instance is then found in $(w+3)^w$ time by considering each value of type $OPT(\Psi, \cdot)$.
	Thus, the running time of the whole algorithm is $|\Psi| \cdot (w+3)^{2w} \cdot n^{\O(1)}$.
	The correctness follows from the discussion.
	This finishes the proof.
\end{proofO}

	\section{Polynomial kernel for \MHVfull}\label{sec:structural-kernel}

In this section, we present a polynomial kernel for \MHV~parameterized by the distance to clique.
This partially answers a question of Misra and Reddy in \cite{Misra2018}, where they also showed \FPT~algorithms for both \MHV~and \MHE~parameterized by this parameter.
\ifshort
\else
We start with the following technical lemma.

\begin{lemmaO}\label{lemma:mhv_clique_linear_kernel}
	\MHVfull~admits a kernel with $\O(|\mathcal{H}(G,p)|+|S|^2)$ vertices, if a clique modulator $S$ of $G$ is given.
\end{lemmaO}
\begin{proofO}
	Let $(G,p,k,S)$ be an instance of \MHV~with a clique modulator $S$ of $G$ given.
	That is, $V(G)=C\sqcup S$, and $C$ induces a clique in $G$.
	Let $h=|\mathcal{H}(G,p)|$ as usual.
	This lemma has much in common with Corollary \ref{cor:mhv_kernel_pot_happy}.
	Here, to obtain a linear dependency on $h$, we exploit the fact that $G$ contains a large clique.
	
	Our kernelization algorithm outputs an instance $(G',p',k)$ of \MHV~with $|V(G')|\le 2|\mathcal{H}(G,p)|+3|S|+\binom{|S|}{2}+2$.
	The graph $G'$ is obtained as an induced graph $G[T]$ of $G$ for some vertex subset $T$.
	Additional two vertices are then introduced in $G'$ to ensure that certain vertices are not happy.
	We now show how the algorithm constructs the set $T$.
	By means of Theorem \ref{thm:mhv_gmc_connection}, the goal of $T$ is to preserve all potentially happy vertices, colors of their neighbours, and paths of length two between them.
	
	Firstly, the algorithm puts all potentially vertices of $(G,p)$ and all vertices of the clique modulator in $T$, i.e.\ $T=\mathcal{H}(G,p) \cup S$ initially.
	Each vertex in $\mathcal{H}(G,p)$ either has no precolored neighbours, or all its neighbours are precolored with the same color.
	For each such vertex that has a precolored neighbour, say $v \in \mathcal{H}_i(G,p)$ for some $i \in [\ell]$, the algorithm chooses any neighbour $u$ of $v$ with $p(u)=i$, and puts it in $T$.
	Now $T$ consists of at most $2h+2|S|$ vertices.
	
	It is left to add vertices in $T$ that preserve paths of length two between potentially happy vertices.
	For each vertex in $S$, say $s \in S$, the algorithm picks any neighbour of $s$ in $C$, and adds it to $T$.
	If $s$ has no neighbours in $C$, the algorithm does nothing.
	Finally, for each pair of vertices in $S$, say $s_1, s_2 \in S$, that has a common neighbour in $G$, the algorithm adds 
	any common neighbour of $s_1$ and $s_2$ to $T$.
	It is easy to see that $|T|\le 2h+3|S|+3\binom{|S|}{2}$ now.
	
	The graph $G'$ is then obtained as the induced graph $G[T]$ of $G$. The precoloring $p'$ is obtained as just a restriction $\left.p\right|_T$ of the precoloring $p$ to the set $T$.
	The only problem behind $(G',p')$ is that $\mathcal{H}_i(G,p)=\mathcal{H}_i(G',p')$ not necessarily holds true for each $i \in [\ell]$: some vertices that are not potentially happy in $(G,p)$ can become potentially happy in $(G',p')$. 
	To overcome that, the algorithm introduces two vertices $t_1, t_2$ to $G'$, connects them by an edge in $G'$, and precolors them with colors $1$ and $2$ respectively, i.e.\ $p'(t_i)=i$.
	Then it connects both $t_1$ and $t_2$ with each vertex in $T \setminus \mathcal{H}(G,p)$ by an edge in $G'$.
	The construction of $(G',p',k)$ is finished.
	
	\begin{claim}
		For the constructed instance $(G',p',k)$ the following holds:
		\begin{enumerate}
			\item $\mathcal{H}(G,p)=\mathcal{H}(G',p')$;
			\item For each $i \in [\ell]$, $\mathcal{H}_i(G,p)=\mathcal{H}_i(G',p')$;
			\item $G^2[\mathcal{H}(G,p)]$ and $(G')^2[\mathcal{H}(G',p')]$ are the same.
		\end{enumerate}
	\end{claim}
	\begin{claimproof}
		For each potentially happy vertex $v$ of $G$, it was preserved in $T$ and at least one of its colored neighbours is preserved in $T$.
		Thus, $\mathcal{H}(G,p)\subseteq\mathcal{H}(G',p')$ and $\mathcal{H}_i(G,p)\subseteq\mathcal{H}_i(G',p')$.
		At the other hand, all vertices in $T$ that are not potentially happy in $G$, are connected with two vertices $t_1$ and $t_2$, that are of different colors in $(G',p')$.
		Moreover, neither $t_1$ nor $t_2$ are potentially happy in $(G',p')$ since they are connected by an edge.
		Therefore, the first two conditions of the claim follows.
		
		To prove the third condition of the claim, note that $$(G')^2[\mathcal{H}(G',p')]=(G')^2[\mathcal{H}(G,p)]=(G[T])^2[\mathcal{H}(G,p)],$$ since vertices $t_1$ and $t_2$ of $G'$ are connected only with vertices that are not potentially happy in $(G,p)$.
		Hence, $t_1$ and $t_2$ cannot contribute to a path of length two between a pair of vertices in $\mathcal{H}(G,p)$.
		Suppose that $(G[T])^2[\mathcal{H}(G,p)]\neq G^2[\mathcal{H}(G,p)]$.
		Then, it is only the case that some path of length two between some non-adjacent vertices $u, v\in \mathcal{H}(G,p)$ is missing in $G[T]$, but they share a common neighbour $w$ in $G$.
		As $w$ is missing in $G[T]$, it is the case that $w \in C$. Otherwise $w \in S$ and the algorithm would include it in the set $T$ initially.
		
		Since $u$ and $v$ are non-adjacent, we may assume without loss of generality that $u \in S$.
		Suppose that $v \in S$.
		Then, $u$ and $v$ are vertices in $S$ that share a common neighbour in $G$.
		But then the algorithm has added at least one their neighbour in $T$, so there should be a path of length two between them in $G[T]$.
		Hence, it is the case that $v \in C$.
		But $w \in C$ also, so $u$ has a neighbour in $C$, say $w' \in C$, that the algorithm has included in $T$.
		$uw'v$ is a path of length two between $u$ and $v$ in $G[T]$, a contradiction.
		Therefore, the third condition of the claim holds.
	\end{claimproof}
	
	From the claim and Theorem \ref{thm:mhv_gmc_connection} follows the lemma.
\end{proofO}

\begin{lemmaO}\label{lemma:mhv_distance_to_clique_kernel}
	\MHVfull~admits a polynomial kernel of size $\O(d^3)$, where $d$ is the size of a given clique modulator $S$.
\end{lemmaO}
\begin{proofO}
	Let $(G, p, k, S)$ be an instance of \MHV~with a clique modulator $S$ of $G$ of size $d$ given.
	As usual, we denote the set of the vertices of the clique by $C$, i.e.\  $V(G)=C\sqcup S$.
	Throughout the proof, we assume that $d\ge 2$, otherwise the instance is trivial.
	
	We present an algorithm that transforms $(G, p, k, S)$ into an instance $(G', p',$ $ k')$ of \MHV~with $|V(G')|=\O(d^3)$.
	The algorithm reduces the number of potentially happy vertices in $(G,p)$ step-by-step.
	To do that, the algorithm makes some potentially happy vertices unhappy.
	As does the algorithm in the proof of Lemma \ref{lemma:mhv_clique_linear_kernel},
	the algorithm introduces two new adjacent vertices $t_1$ and $t_2$ to $G$ and precolors them with colors $1$ and $2$ respectively.
	Then, to make a potentially happy vertex $v$ unhappy, the algorithm connects $t_1$ with $v$ and $t_2$ with $v$ by introducing two new edges.
	Note that this operation is done in polynomial time and strictly decreases the number of potentially happy vertices.
	$t_1$ and $t_2$ are introduced to $G$ only once, and after the introduction $S$ may no longer be a clique modulator of $G$.
	To fix that, we say that algorithm extends $S$ with $t_1$ and $t_2$.
	Thus, $d$ increases by a constant value of two.
	We further assume that $G$ contains the vertices $t_1$ and $t_2$, $S$ is the extended clique modulator and $d=|S|$ is equal to its size.
	
	When the number of potentially happy vertices becomes $\O(d^3)$, the algorithm continues following the Lemma \ref{lemma:mhv_clique_linear_kernel}.
	We now show how the algorithm achieves $\O(d^3)$ potentially happy vertices.
	
	Firstly, the algorithm ensures that each color in $[\ell]$ is presented at least once in $p$. It applies the following reduction rule exhaustively.
	
	\begin{rrule}\label{rrule:cluster_color_not_present}
		If there is a color $i \in [\ell]$ that is not presented in $p$, $p^{-1}(i)=\emptyset$, decrease color numbers in $\{i+1, i+2, \ldots \ell \}$ by $1$ and decrease $\ell$ by $1$.
	\end{rrule}
	
	Then, the following claim allows to deal with the case when the number of colors in $(G,p)$ is sufficiently large.
	\ifshort
	\addtocounter{claim}{8}
	\fi
	\begin{claimO}
		If $\ell > d+1$, then either there is a color $i \in [\ell]$ with $p^{-1}(i)=\emptyset$, or only the vertices of $S$ can be happy in $(G, p)$, i.e.\ $S \supseteq \mathcal{H}(G,p)$.
	\end{claimO}
	\begin{claimproofO}
		Suppose $\ell > d+1$ and each color in $[\ell]$ is presented at least once in $p$.
		Then there is a sequence of distinct vertices $v_1, v_2, \ldots, v_\ell$ with $p(v_i)=i$.
		Since $\ell \ge d+2$ and $|S|=d$, at least two vertices in the sequence are from the clique $C$.
		Colors of these two vertices are distinct, so no vertex in the clique can be happy with respect to any coloring of $(G,p)$.
	\end{claimproofO}
	
	The claim shows that after the exhaustive application of Reduction rule \ref{rrule:cluster_color_not_present}, if $\ell$ is large, then the number of potentially happy vertices in that case is at most $d$.
	The following part of the algorithm is dealing with the case when $\ell \le d+1$.
	
	Then, the algorithm finds sets $C_1, C_2, \ldots, C_\ell$, where \ifshort $ \else $$ \fi C_i=C\cap \mathcal{H}_i(G,p) \ifshort $ \else $$ \fi is the set of potentially happy vertices in $C$ that are either precolored with color $i$ or has a neighbour precolored with color $i$ in $G$.
	Also, algorithm finds a set $C_0$ of potentially happy vertices in $C$ that are not precolored and have no precolored neighbour in $(G,p)$, $C_0=\mathcal{H}(G,p)\setminus C_1 \setminus C_2 \setminus \ldots \setminus C_\ell$.
	The sequence $C_0, C_1, \ldots, C_\ell$ is found in polynomial time.
	If all sets in the sequence have size at most $d$, then the number of potentially happy vertices in $(G,p)$ is at most $\ell \cdot d \le (d+1) \cdot d$.
	The following claim helps to deal with the other case.
	
	\begin{claimO}\label{claim:cluster_big_color}
		If $|C_0| + |C_i| > d$ for some $i \in [\ell]$, then in any optimal coloring $c$ of $(G,p)$, all vertices of the clique are colored with the same color by $c$, i.e.\ $|c(C)|=1$.
	\end{claimO}
	\begin{claimproofO}
		If $|c(C)| \ge 2$, only vertices in $S$ can be happy with respect to $c$.
		Since $|S|=d$, at most $d$ vertices can be happy with respect to $c$ in $(G,p)$. However, since $|C_0|+|C_i|>d$, a trivial extension of $p$ with the color $i$ yields at least $d+1$ happy vertices.
	\end{claimproofO}
	
	The algorithm then applies the following reduction rule, that gets rid of potentially happy vertices in $C$ that are never happy in any optimal coloring of $(G,p)$. 
	
	\begin{rrule}\label{rrule:cluster_empty_cols}
		If there exists $i \in [\ell]$ with $|C_i| + d < \max\limits_{j=1}^{\ell} |C_j|$, make all vertices in $C_i$ unhappy.
	\end{rrule}
	\begin{claimO}
		Reduction rule \ref{rrule:cluster_empty_cols} is safe.
	\end{claimO}
	\begin{claimproofO}
		Note that $|C_j|\ge d+1$. Then, by Claim \ref{claim:cluster_big_color}, in any optimal coloring $c$, $c(C)=1$.
		Suppose the reduction rule is not safe.
		Then, there is a optimal coloring $c$ with $c(C)=\{i\}$.
		$c$ yields at most $|C_0|+|C_i|+d$ happy vertices in $(G,p)$.
		At the other hand, a trivial extension of $p$ with color $j$ yields at least $|C_0|+|C_j|>|C_0|+|C_i|+d$ happy vertices.
		Hence, $c$ is not an optimal coloring.
		Moreover, for any optimal coloring of $c$, no vertex in $C_i$ is happy with respect to $c$.
	\end{claimproofO}
	
	After a single application of Reduction rule \ref{rrule:cluster_empty_cols}, $C_i$ becomes empty.
	The algorithm applies the rule exhaustively.
	Denote the set of colors corresponding to non-empty sets in the sequence $C_1, C_2, \ldots, C_\ell$ by $L=\{i \mid C_i \neq \emptyset \}$.
	
	At the next step, the algorithm deals with non-precolored vertices in $S$.
	The obstacle behind non-precolored vertices in $S$ is that we can not be sure about their color in an optimal coloring.
	Depending on the color of the clique, certain colorings of certain non-precolored vertices in $S$ can make some vertices in the clique not happy.
	The following claim helps in reducing the number of clique neighbours for the non-precolored vertices.
	
	\begin{claimO}\label{claim:cluster_many_neighbours}
		In any optimal coloring $c$ of $(G,p)$, for any $i \in L$ and any non-precolored vertex $v \in S$, if $|N(v)\cap C_i| \ge d$ and $c(C)=\{i\}$, then $c(v)=i$. 
		Also, if $i=0$ and $|N(v) \cap C_0| \ge d$ and $c(C)=\{j\}$, then $c(v)=j$.
	\end{claimO}
	\begin{claimproofO}
		Suppose $i \in L$ and $c$ is an optimal coloring with $c(C)=\{i\}$, but there is a vertex $v \in S$ with that has at least $d$ neighbours in $C_i$, and $c(v)\neq i$.
		Note that only vertices in $C_0 \cup C_i \cup S$ can be happy with respect to $c$.
		
		Denote $$B_i=\{u \mid u \in S,\; |N(u)\cap (C_0\cup C_i)|>0,\; c(u)\neq i\}.$$
		That is, $B_i$ is a set of neighbours of $C_0\cup C_i$ that are colored with a color different from $i$ in $c$.
		$B_i$ is not empty and no vertex in $B_i$ is happy with respect to $c$.
		Also, since $v \in B_i$, $|N(B_i)\cap C_i|\ge d$.
		
		Construct a coloring $c'$ by changing colors of all vertices in $B_i$ to $i$ in $c$.
		That is, $c'(B_i)=\{i\}$, but $c'(u)=c(u)$ for each $u \in V(G) \setminus B_i$.
		After such change, some vertices in $S \setminus B_i$ that are happy with respect to $c$, become not happy with respect to $c'$.
		On the other hand, all vertices in $N(B_i) \cap (C_0\cup C_i)$ (all of them are not happy with respect to $c$) become happy with respect to $c'$.
		No other vertex is influenced by the change.
		Since $|S \setminus B_i| \le d-1$ and $|N(B_i) \cap C_i| \ge d$, $c'$ yields at least one happy vertex more than $c$ does.
		A contradiction with the optimality of $c$.
		
		The case $i=0$ is handled in the same way.
	\end{claimproofO}
	
	The claim results in the following reduction rule.
	
	\begin{rrule}\label{rrule:cluster_remove_edge}
		If there is a non-precolored vertex $v \in S$ with $|N(v) \cap C_i| > d$ for some $i \in \{0\} \cup L$, take any $u \in N(v) \cap C_i$ and remove the edge $vu$ from $G$.
	\end{rrule}
	\begin{claimO}
		Reduction rule \ref{rrule:cluster_remove_edge} is safe.
	\end{claimO}
	\begin{claimproofO}
		Suppose $(G,p,k)$ and $(G\setminus vu,p,k)$ are not equivalent instances.
		Note that for any coloring $c$, $\mathcal{H}(G,c)\subseteq \mathcal{H}(G\setminus vu, c)$.
		Instances are not equivalent, so there is an optimal coloring $c$ of $(G\setminus vu,p)$ that yields at least $k$ happy vertices in $(G\setminus vu,p)$.
		But since $(G,p,k)$ is a no-instance, $c$ yields at most $k-1$ happy vertices in $(G,p)$.
		Thus, the edge $vu$ changes the happiness of at least one of $u$ and $v$ with respect to $c$. In particular, $c(u)\neq c(v)$.
		
		At the other hand, $|N_G(v)\cap C_i|>d$, hence $|C_i|>d$ and $|N_{G\setminus vu}(v)\cap C_i|\ge d$.
		By Claim \ref{claim:cluster_big_color}, $c(C)=\{j\}$ for some $j \in L$.
		Hence, $c(u)=j$, but $c(v)\neq j$.
		Suppose that $i=0$.
		By Claim \ref{claim:cluster_many_neighbours} applied to $c$ and $(G\setminus vu,p)$, it holds that $c(v)=j$, a contradiction.
		
		Then it is the case that $i \neq 0$.
		Since $c(C)=\{j\}$ and $c(v)\neq j$, $v$ is not happy in $G\setminus vu$ with respect to $c$.
		Then the edge $vu$ changes the happiness of $u$, so $u$ is happy with respect to $c$ in $G\setminus vu$. 
		Since $i\neq 0$ and $u\in C_i$ is a happy vertex, $c(u)=i$, so $i=j$.
		But $|N_{G\setminus vu}(v)\cap C_i|\ge d$ and $c(C)=\{i\}$, and from Claim \ref{claim:cluster_many_neighbours} follows that $c(v)=i$. This contradiction finishes the proof.
	\end{claimproofO}
	
	The algorithm applies Reduction rule \ref{rrule:cluster_remove_edge} exhaustively.
	Then, when it got rid of non-precolored vertices with many neighbours, it gets rid of clique vertices with no non-precolored neighbours in $S$.
	This is formulated in the following two reduction rules, that the algorithm applies exhaustively.
	
	\begin{rrule}\label{rrule:cluster_precs_uncols}
		If $|C_0| > d+1$ and $C_0$ contains a vertex that has no neighbours in $S$, make that vertex unhappy and decrease $k$ by $1$.
	\end{rrule}
	\begin{claimO}
		Reduction rule \ref{rrule:cluster_precs_uncols} is safe.
	\end{claimO}
	\begin{claimproofO}
		Let $(G,p,k)$ be an instance with
		$|C_0| > d+1$.
		Let $v \in C_0$ be a vertex that has no neighbours in $S$.
		Let $(G',p,k-1)$ be the instance obtained after a single application of the reduction rule.
		$G'$ differs from $G$ only in two edges that ensure unhappiness of $v$.
		Obviously, if there is an optimal coloring of $(G,p)$ that yields at least $k$ happy vertices in $(G,p)$, the same coloring yields at least $k-1$ happy vertices in $(G',p)$.
		
		Take now an optimal coloring $c'$ of $(G',p)$. $|C_0 \setminus \{v\}|\ge d+1$, and again by Claim \ref{claim:cluster_big_color} (applied now to $(G',p')$), $|c'(C)|=1$.
		Apply $c'$ to $(G,p)$.
		$|c'(C)|=1$, so $v$ is happy with respect to $c'$ in $(G,p)$.
		Thus, $c'$ yields all the happy vertices in $(G,p)$ that it does in $(G',p')$, and also the vertex $v$.
		Therefore, if an optimal coloring of $(G',p')$ yields at least $k-1$ happy vertices, there is an optimal coloring of $(G,p)$ that yields at least $k$ vertices.
		The safeness of the reduction rule follows.
	\end{claimproofO}
	
	The following reduction rule is of the same nature, but is a bit more complicated.
	
	\begin{rrule}\label{rrule:cluster_precs_cols}
		If for each $i \in L$, $|C_i| > d+1$ and $C_i$ contains a vertex that has only precolored neighbours in $S$, do the following.
		For each $i \in L$, make one such vertex in $C_i$ unhappy. Decrease $k$ by $1$.
	\end{rrule}
	\begin{claimO}
		Reduction rule \ref{rrule:cluster_precs_cols} is safe.
	\end{claimO}
	\begin{claimproofO}
		Let $(G,p,k)$ be the instance before an application of the reduction rule.
		$|C_i|> d+1$ holds for each $i \in [\ell]$.
		For each $i \in [\ell]$, let $v_i \in C_i$ be the vertex in $C_i$ that has only precolored neighbours in $S$.
		Let $(G',p,k-1)$ be the instance after the application of the reduction rule.
		
		We show that if $(G,p,k)$ is a yes-instance, then $(G',p',k-1)$ is a yes-instance.
		Let $c$ be an optimal coloring of $(G,p)$ yielding at least $k$ happy vertices.
		By Claim \ref{claim:cluster_big_color}, $c(C)=\{i\}$, and $i \in L$.
		Note that $v_i$ is happy with respect to $c$, since all vertices of the clique are colored with color $i$ and all vertices in $N(v_i)\cap S$ are precolored with the color $i$.
		Moreover, for each $j \in [\ell]\setminus \{i\}$, no vertex in $C_j$ is happy with respect to $c$.
		Thus, $c$ yields all the same happy vertices in $(G',p)$ as it does in $(G,p)$, except for the single vertex $v_i$.
		Hence, $c'$ yields at least $k-1$ happy vertices in $(G',p)$.
		
		To prove in the other direction, take an optimal coloring $c'$ of $(G',p)$ that yields at least $k-1$ happy vertices.
		$|C_i\setminus \{v_i\}| \ge d+1$ for each $i\in [\ell]$, so apply Claim \ref{claim:cluster_big_color} to $(G',p')$ and get that $|c'(C)|=\{i\}$ for some $i\in L$.
		$v_i$ is not happy in $(G',p')$ with respect to $c'$, but it is happy in $(G,p)$ with respect to $c'$.
		All other happy vertices remain the same.
		Hence, $c'$ yields at least $k$ happy vertices in $(G,p)$.
		This finishes the proof.
	\end{claimproofO}
	
	We finally claim that the number of remaining potentially happy vertices is $\O(d^3)$.
	
	\begin{claimO}
		After the exhaustive application of the reduction rules, the clique $C$ contains at most $d^2+ d \cdot (d+1)^2$ potentially happy vertices, i.e.\ $|C \cap \mathcal{H}(G,p)| \le d^2+ d \cdot (d+1)^2$.
	\end{claimO}
	\begin{claimproofO}
		Observe that $|C \cap \mathcal{H}(G,p)|=|C_0\cup C_1 \cup C_2 \cup \ldots \cup C_\ell|=|C_0|+\sum_{i\in L} |C_i|$.
		We bound $|C_0|$ and $\sum_{i\in L} |C_i|$ separately.
		
		Each non-precolored vertex in $S$ has at most $d$ neighbours in $C_0$, otherwise Reduction rule \ref{rrule:cluster_remove_edge} could be applied.
		This contributes to at most $d^2$ such vertices in $C_0$.
		And since Reduction rule \ref{rrule:cluster_precs_uncols} cannot be applied, $C_0$ either contains no other vertices or consists of at most $d+1$ vertices in total.
		Hence, $|C_0|\le \max\{d^2, d+1\}=d^2$.
		
		We now bound $\sum_{i\in L} |C_i|$.
		We suppose that $L$ is not empty, otherwise \linebreak $\sum_{i\in L} |C_i|=0$.
		Since Reduction rule \ref{rrule:cluster_precs_cols} cannot be applied, there is either a set $C_i$ with $|C_i|\le d+1$, or there is a set $C_j$ that consists only of neighbours of non-precolored vertices in $S$.
		In the latter case, $|C_j| \le d^2$, otherwise Reduction rule \ref{rrule:cluster_remove_edge} could be applied.
		In any case, there is a set $C_k$ of size at most $d^2$, $k \in L$.
		
		Recall that Reduction rule \ref{rrule:cluster_empty_cols} ensures that $|C_i| \le |C_k|+d$ for each $i \in L$.
		Since it was applied exhaustively, we get that $\sum_{i\in L}|C_i|\le \ell \cdot (|C_k|+d) \le (d+1) \cdot (d^2+d)=d\cdot (d+1)^2$.
	\end{claimproofO}
	
	The claim shows that the number of potentially happy vertices is bounded and Lemma \ref{lemma:mhv_clique_linear_kernel} can be applied.
	The proof is finished.
\end{proofO}

The lemmata above require that a clique modulator of $G$ is given as an input.
This is not that necessary, since the distance to clique number is $2$-approximable in polynomial time.

\begin{lemmaO}\label{lemma:distance_to_clique_approximable}
	There is a polynomial-time algorithm that finds a clique modulator of a given $G$ consisting of at most $2d$ vertices, where $d$ is the size of minimum clique modulator of $G$.
\end{lemmaO}
\begin{proofO}
	Observe that a clique modulator $S$ of $G$ is a vertex cover of its complement $\overline{G}$.
	And vice versa, a vertex cover of $\overline{G}$ is a clique modulator in $G$.
	Thus, $d$ is the size of minimum vertex cover of $\overline{G}$.
	Take a well-known $2$-approximation algorithm for vertex cover by Gavril and Yannakakis \cite{0131524623}, and apply it to $\overline{G}$.
	Resulting vertex cover of size at most $2d$ is a clique modulator of $G$.
\end{proofO}

We combine algorithms of Lemma \ref{lemma:distance_to_clique_approximable} and Lemma \ref{lemma:mhv_distance_to_clique_kernel} to finally obtain the following result.
\fi
\begin{theoremO}
	\MHVfull~admits a kernel with $\O(d^3)$ vertices, where $d$ is the distance to clique parameter, and the parameter and a clique modulator of $G$ are not given explicitly.
\end{theoremO}

	\nocite{*}
	\bibliography{wg-cited}

\end{document}